\documentclass[letterpaper,11pt]{article}
\usepackage{newtxtext} %
\usepackage[margin=1in]{geometry}
\usepackage{setspace}
\usepackage{graphicx}
\usepackage{epsfig}

\usepackage{microtype} %
\usepackage{array} %

\usepackage{amssymb}
\usepackage{amsmath}
\usepackage{amsthm}
\usepackage{mathtools} %
\usepackage{stmaryrd} %
\usepackage{stackrel} %
\usepackage{latexsym} %

\usepackage{tikz}
\usetikzlibrary{arrows}
\usetikzlibrary{arrows.meta}
\usetikzlibrary{decorations.pathreplacing}
\usetikzlibrary{cd}
\usepackage[all]{xy} %

\usepackage{algorithm}
\usepackage[noend]{algpseudocode}
\algnewcommand{\LineComment}[1]{\Statex\hspace{\algorithmicindent}\(\triangleright\) #1}
\algnewcommand\algorithmicforeach{\textbf{for each}}
\algdef{S}[FOR]{ForEach}[1]{\algorithmicforeach\ #1\ \algorithmicdo}
\usepackage{etoolbox} %

\usepackage{caption}
\usepackage{subcaption}

\usepackage{makecell} %
\usepackage{booktabs} %
\usepackage{multirow} %
\usepackage{tablefootnote} %

\usepackage[hidelinks]{hyperref} %
\usepackage[symbol]{footmisc}
\usepackage{makeidx} %
\usepackage{url} %
\usepackage{color}
\usepackage{xcolor}
\usepackage{xspace} %
\usepackage[normalem]{ulem} %
\usepackage{soul} %
\usepackage{ifpdf} %
\usepackage{extarrows}

\makeatletter
\expandafter\patchcmd\csname\string\algorithmic\endcsname{\itemsep\z@}{\itemsep=0.25ex}{}{}
\makeatother

\newcounter{usesmallsep}
\ifnum\the\value{usesmallsep}=1

    \newlength{\myitemsep}
    \newlength{\mytopsep}
    \usepackage{enumitem}
    \setlist[itemize]{leftmargin=\parindent,parsep=\parskip,
      listparindent=\parindent,itemsep=\myitemsep,topsep=\myitemsep}
    \setlist[enumerate]{leftmargin=\parindent,parsep=\parskip,
      listparindent=\parindent,itemsep=\myitemsep,,topsep=\myitemsep}
    \setlist[description]{font=\bfseries,leftmargin=\parindent,parsep=\parskip,
      listparindent=\parindent,itemsep=\myitemsep,topsep=\myitemsep}
    
    \newlength{\mypartitlesep}
    \usepackage[explicit]{titlesec}
    \titlespacing{\paragraph}{0pt}{\mypartitlesep}{\mypartitlesep}
    
    \newlength{\mythmsep}
    \newtheoremstyle{mythmstyle}
      {\mythmsep} %
      {\mythmsep} %
      {\itshape} %
      {} %
      {\bfseries} %
      {.} %
      {.5em} %
      {} %
    
    \newtheoremstyle{mydefstyle}
      {\mythmsep} %
      {\mythmsep} %
      {} %
      {} %
      {\bfseries} %
      {.} %
      {.5em} %
      {} %
    
    \theoremstyle{mythmstyle}
        \newtheorem{theorem}{Theorem}
        \newtheorem{proposition}[theorem]{Proposition}
        \newtheorem{lemma}[theorem]{Lemma}
        \newtheorem{corollary}[theorem]{Corollary}
        \newtheorem{fact}[theorem]{Fact}
        \newtheorem*{fact*}{Fact}
    \theoremstyle{mydefstyle}
        \newtheorem{definition}{Definition}
        \newtheorem{problem}{Problem}
        \newtheorem{assumption}{Assumption}
        \newtheorem{remark}{Remark}
        \newtheorem{algr}[algorithm]{Algorithm}

    \newenvironment{proof}
        {\vspace{-0.9em}\begin{proof}}
        {\end{proof}\vspace{-0.4em}}

\else

    \theoremstyle{plain}
        \newtheorem{theorem}{Theorem}
        \newtheorem{proposition}[theorem]{Proposition}

        \newtheorem*{algr*}{Algorithm}
    \theoremstyle{definition}
        \newtheorem{definition}[theorem]{Definition}

    \usepackage{enumitem}
    \setlist[itemize]{leftmargin=\parindent}
    \setlist[enumerate]{leftmargin=\parindent}
    \setlist[description]{font=\bfseries,leftmargin=\parindent}

\fi

\newcommand{\Hm}{\mathsf{H}}

\newcommand{\Chn}{\mathsf{C}}
\newcommand{\Cyc}{\mathsf{Z}}
\newcommand{\Bnd}{\mathsf{B}}

\newcommand{\Real}{\mathbb{R}}

\newcommand{\rank}{\mathsf{rank}\,}
\renewcommand{\ker}{\mathsf{ker}}

\newcommand{\Pers}{\mathsf{Pers}}

\newcommand{\pinds}{\mathsf{P}}
\newcommand{\ninds}{\mathsf{N}}

\renewcommand{\bar}[1]{\overline{#1}}

\newcommand{\inv}{^{-1}}

\newcommand{\lbarrowspace}{\;}

\let\leftrightarrowsp\lrarrowsp

\newcommand{\incto}{\hookrightarrow}
\newcommand{\inctosp}[1]{\xhookrightarrow{\lbarrowspace#1\lbarrowspace}}
\newcommand{\bakincto}{\hookleftarrow}
\newcommand{\bakinctosp}[1]{\xhookleftarrow{\lbarrowspace#1\lbarrowspace}}

\newcommand{\given}{\,|\,}
\newcommand{\Set}[1]{\{#1\}}

\newcommand{\Max}{\mathsf{max}}

\let\emptyset\varnothing

\let\union\cup
\let\bigunion\bigcup

\newcommand{\Bcal}{\mathcal{B}}

\newcommand{\Ecal}{\mathcal{E}}
\newcommand{\Dcal}{\mathcal{D}}
\newcommand{\Fcal}{\mathcal{F}}

\newcommand{\Ical}{\mathcal{I}}

\newcommand{\Rcal}{\mathcal{R}}

\newcommand{\Ucal}{\mathcal{U}}

\newcommand{\Zbb}{\mathbb{Z}}

\newcommand{\bG}{\beta}
\newcommand{\dG}{\delta}
\newcommand{\DG}{\Delta}

\newcommand{\GG}{\Gamma}

\newcommand{\LG}{\Lambda}

\newcommand{\oG}{\omega}

\newcommand{\sG}{\sigma}

\newcommand{\tG}{\tau}

\newcommand{\Dim}{p}

\newcommand{\birth}{b}
\newcommand{\death}{d}
\newcommand{\filtcnt}{m}
\newcommand{\simpcnt}{n}
\newcommand{\bles}{\prec_{\mathsf{b}}}
\newcommand{\dles}{\prec_{\mathsf{d}}}
\newcommand{\iles}{\prec}
\newcommand{\brepsum}{\oplus_{\mathsf{b}}}
\newcommand{\drepsum}{\oplus_{\mathsf{d}}}
\newcommand{\repsum}{\oplus}
\newcommand{\cyc}{z}
\newcommand{\chn}{c}
\newcommand{\fsimp}[2]{\sigma_{#2}}
\newcommand{\rseq}{\zeta}
\newcommand{\rconcat}{\mathbin\Vert}
\newcommand{\rsprefix}[1]{{[:\!#1]}}
\newcommand{\rssuffix}[1]{{[#1\!:]}}
\newcommand{\dpc}{\Dcal}
\newcommand{\pset}{P}
\newcommand{\dpccnt}{s}
\newcommand{\posf}{D}

\newcommand{\ef}{\Ecal}
\newcommand{\hatfsimp}{\hat{\sigma}}

\makeatletter
\newcommand*{\da@rightarrow}{\mathchar"0\hexnumber@\symAMSa 4B }
\newcommand*{\da@leftarrow}{\mathchar"0\hexnumber@\symAMSa 4C }
\newcommand*{\xdashrightarrow}[2][]{%
  \mathrel{%
    \mathpalette{\da@xarrow{#1}{#2}{}\da@rightarrow{\;}{}}{}%
  }%
}
\newcommand{\xdashleftarrow}[2][]{%
  \mathrel{%
    \mathpalette{\da@xarrow{#1}{#2}\da@leftarrow{}{}{\;}}{}%
  }%
}
\newcommand{\xdashleftrightarrow}[2][]{%
  \mathrel{%
    \mathpalette{\da@xarrow{#1}{#2}\da@leftarrow\da@rightarrow{}{}}{}%
  }%
}
\newcommand*{\da@xarrow}[7]{%
  \sbox0{$\ifx#7\scriptstyle\scriptscriptstyle\else\scriptstyle\fi#5#1#6\m@th$}%
  \sbox2{$\ifx#7\scriptstyle\scriptscriptstyle\else\scriptstyle\fi#5#2#6\m@th$}%
  \sbox4{$#7\dabar@\m@th$}%
  \dimen@=\wd0 %
  \ifdim\wd2 >\dimen@
    \dimen@=\wd2 %
  \fi
  \count@=2 %
  \def\da@bars{\dabar@\dabar@}%
  \@whiledim\count@\wd4<\dimen@\do{%
    \advance\count@\@ne
    \expandafter\def\expandafter\da@bars\expandafter{%
      \da@bars
      \dabar@ 
    }%
  }%
  \mathrel{#3}%
  \mathrel{%
    \mathop{\da@bars}\limits
    \ifx\\#1\\%
    \else
      _{\copy0}%
    \fi
    \ifx\\#2\\%
    \else
      ^{\copy2}%
    \fi
  }%
  \mathrel{#4}%
  \!\!
}
\makeatother

\let\rseqlrarr\xdashleftrightarrow
\let\rseqlarr\xdashleftarrow
\let\rseqrarr\xdashrightarrow

\newcounter{desccounter}

\newcommand{\descitem}[2]{\stepcounter{Item}\refstepcounter{desccounter}\item[#1 \Alph{desccounter} #2]}

\newcommand{\defemph}[1]{{\textit{#1}}}

\newcommand{\cancel}[1]

\begin{document}

\title{Updating Barcodes and Representatives for Zigzag Persistence\thanks{This research is partially supported by NSF grant CCF 2049010.}}

\author{Tamal K. Dey\thanks{Department of Computer Science, Purdue University. \texttt{tamaldey@purdue.edu}}
\and Tao Hou\thanks{School of Computing, DePaul University. \texttt{thou1@depaul.edu}}
}

\date{}

\maketitle
\thispagestyle{empty}

\begin{abstract}
Computing persistence over changing filtrations give rise to a
stack of 2D persistence diagrams where the birth-death points 
are connected by the so-called `vines'~\cite{cohen2006vines}. We consider
computing these vines over changing filtrations for zigzag persistence.
We observe that eight atomic operations are sufficient for changing
one zigzag filtration to another and provide update
algorithms for each of them. Six of these operations that have
some 
analogues to one or multiple \emph{transpositions} in the non-zigzag case can be executed as efficiently
as their non-zigzag counterparts. This approach
takes advantage of a recently discovered algorithm for
computing zigzag barcodes~\cite{dey2022fast} by converting a zigzag filtration to a
non-zigzag one and then connecting barcodes of the two with a bijection. 
The remaining two atomic operations do not have a strict analogue in the non-zigzag case.
For them, we propose algorithms based on explicit maintenance of representatives (homology cycles) which can be useful in their own rights for applications requiring explicit updates of representatives.

\end{abstract}

\newpage
\setcounter{page}{1}

\section{Introduction}
\label{sec:intro}
Computation of the persistence diagram (PD) from a given filtration has turned out to
be a central task in topological data analysis. Such a filtration usually represents
a nested sequence of sublevel sets of a function. In scenarios where the function changes, the filtration and hence the PD may also change. The authors in~\cite{cohen2006vines} provided
an efficient algorithm for updating the PD over an atomic operation which
\emph{transposes} two consecutive simplex additions in the filtration. 
Using this atomic operation repeatedly, one can connect a series of filtrations obtained from a time-varying function with the so-called structure of \emph{vineyard}. 
The authors~\cite{cohen2006vines} showed that the update in PD due to the atomic transposition can be computed in $O(n)$ time if $n$ simplices constitute the filtration. In this paper, we extend this result to zigzag filtrations.
Specifically, we identify \emph{eight atomic operations}
necessary for any zigzag filtration
to transform to any other,
including four that are analogues of transpositions in the non-zigzag case.

Compared to the non-zigzag case, computing the PD (also called the barcode) from a zigzag filtration is itself more complicated. This complication naturally carries over to the task of
updating PDs for changing zigzag filtrations. One main difficulty stems
from the fact that, unlike in the non-zigzag case, it seemed necessary to pay extra cost
in bookkeeping
\emph{representatives} for the bars while computing zigzag barcodes.
The known algorithms by Maria and Oudot~\cite{maria2014zigzag,maria2016computing} (see also~\cite{maria2019discrete}),
Carlsson et al.~\cite{carlsson2009zigzag-realvalue}, and Milosavljevi{\'c} et al.~\cite{milosavljevic2011zigzag} for computing zigzag persistence implicitly or explicitly
maintain these representatives. Naturally, any attempt to adapt these algorithms
to changing filtrations faces the difficulty of updating the representatives efficiently
over the atomic operations. It is by no means obvious how to carry out these updates for
representatives efficiently, let alone avoid them. 

In this paper, we show that, out of the eight atomic operations, we can execute six without
maintaining representatives explicitly
by drawing upon some relations/analogies to the non-zigzag case.
The two remaining
operations whose non-zigzag analogues do not even exist need explicit
maintenance of representatives
due to change in adjacencies of
the cells (see Section~\ref{sec:adj-change}). For the first six operations, we take advantage
of a recently discovered algorithm~\cite{dey2022fast} 
for zigzag persistence
that first
converts a zigzag filtration to a non-zigzag one and then connects barcodes
of the two with a bijection. As shown in~\cite{dey2022fast}, this algorithm called \textsc{FastZigzag} runs quite efficiently in practice because it avoids 
maintaining representatives altogether. 
We summarize our algorithmic results for the operations as follows (see also Table~\ref{tab:complexities} in Section~\ref{sec:overview}):
\begin{itemize}
    \item Four of the eight operations are 
\emph{switches}~\cite{carlsson2010zigzag,carlsson2019parametrized,carlsson2009zigzag-realvalue,maria2014zigzag,maria2016computing,oudot2015zigzag}
which are equivalents of transpositions~\cite{cohen2006vines} in the non-zigzag case. They take
constant or linear time for updates by utilizing the \textsc{FastZigzag} algorithm. 
    \item The other four operations entail
`\emph{expanding}'~\cite{maria2014zigzag,maria2016computing} or `\emph{contracting}' a zigzag filtration locally 
whose equivalents for non-zigzag filtrations have not been considered.
\begin{itemize}
    \item Among them,
two operations (the \emph{inward} expansion and contraction) 
can be related to `expanding' or `contracting' a non-zigzag (standard) filtration by a simplex. One may execute such operations in the non-zigzag case by $O(n)$ transpositions incurring a cost of $O(n^2)$.\footnote{An `expansion' on a non-zigzag filtration can be thought of as inserting a simplex $\sG$ in the middle
of the filtration. The update can be done via inserting $\sG$ to the end
of the filtration and then performing transpositions that
bring $\sG$ to the right position.
A `contraction' on a non-zigzag filtration has the reverse process.}
For these two operations in the zigzag case, we
can still take advantage of the \textsc{FastZigzag} algorithm to have a quadratic time complexity.
    \item The remaining two operations (the \emph{outward} expansion and contraction) are the costliest
which have no direct analogues in the non-zigzag case.
The update algorithms for theses two operations 
require explicit maintenance of representatives
and take cubic time, 
which seems not to be saving time compared to computing the barcodes from scratch~\cite{carlsson2009zigzag-realvalue,dey2022fast,maria2014zigzag,maria2016computing,maria2019discrete}. However, an application
may demand explicit maintenance of the representatives where computing barcodes
from scratch does not help
(see Appendix~\ref{sec:zzup-app} for applications of the representative maintenance). Moreover, our experiment in Section~\ref{sec:dpc-timing} shows that
computing barcodes by our representative-based update algorithms indeed takes less time in practice than computing them afresh for each filtration. Of course, maintaining representatives for one operation requires doing so for every operation. We thereby present an efficient algorithm for explicit maintenance of representatives for every atomic operation.
\end{itemize}

\end{itemize}

In a nutshell, if an application requires only a subset of the first six operations,
barcodes can be updated
as efficiently as in the non-zigzag case. However,
if an application requires explicit maintenance of representatives over 
the operations, or if it requires the last two operations, we pay an extra price.

To motivate our work, we mention in Appendix~\ref{sec:zzup-app} some potential
applications of the update operations/algorithms presented in this paper
to dynamic point clouds and multiparameter (zigzag) persistence.

\section{Preliminaries}

A {\it zigzag filtration} (or simply {\it filtration})
is a sequence of simplicial complexes 
\begin{equation}
\label{eqn:prelim-filt}
\Fcal: K_0 \leftrightarrow K_1 \leftrightarrow 
\cdots \leftrightarrow K_\filtcnt,
\end{equation}
in which each
$K_i\leftrightarrow K_{i+1}$ is either a forward inclusion $K_i\incto K_{i+1}$
or a backward inclusion $K_i\bakincto K_{i+1}$.
Taking the $\Dim$-th homology $\Hm_\Dim$,
we derive a {\it zigzag module}
\[\Hm_\Dim(\Fcal): 
\Hm_\Dim(K_0) 
\leftrightarrow
\Hm_\Dim(K_1) 
\leftrightarrow
\cdots 
\leftrightarrow
\Hm_\Dim(K_\filtcnt), \]
in which
each $\Hm_\Dim(K_i)\leftrightarrow \Hm_\Dim(K_{i+1})$
is a linear map induced by inclusion.
In this paper, we take the coefficient $\Zbb_2$ for $\Hm_\Dim$
and thereby treat chains in the chain group (denoted $\Chn_\Dim$)
and cycles in the cycle group (denoted $\Cyc_\Dim$)
as sets of simplices.
The zigzag module $\Hm_\Dim(\Fcal)$
has a decomposition~\cite{carlsson2010zigzag,Gabriel72} of the form
$\Hm_\Dim(\Fcal)\simeq\bigoplus_{k\in\LG}\Ical^{[\birth_k,\death_k]}$,
in which each $\Ical^{[\birth_k,\death_k]}$
is an
{\it interval module} over the interval $[\birth_k,\death_k]\subseteq\Set{0,\ldots,\filtcnt}$.
The (multi-)set of intervals
$\Pers_\Dim(\Fcal):=\Set{[\birth_k,\death_k]\given k\in\LG}$
is an invariant of $\Hm_\Dim(\Fcal)$
and is called the $\Dim$-th {\it zigzag barcode} (or simply {\it barcode}) of $\Fcal$.
Each interval in $\Pers_\Dim(\Fcal)$
is called a $\Dim$-th \emph{persistence interval}.
We usually consider the homology $\Hm_*$ in all dimensions
and take the zigzag module $\Hm_*(\Fcal)$,
for which we have $\Pers_*(\Fcal)=\bigsqcup_{\Dim\geq 0}\Pers_\Dim(\Fcal)$.
In this paper,
sometimes a filtration may have nonconsecutive indices
on the complexes (i.e., some indices are skipped);
notice that the barcode is still well-defined.

An inclusion in a filtration is called \emph{simplex-wise}
if it is an addition or deletion of a single simplex $\sG$,
which we sometimes denote as, e.g., $K_i\leftrightarrowsp{\sG}K_{i+1}$.
A filtration is called \emph{simplex-wise} if it contains only simplex-wise inclusions.
For computational purposes, we especially focus on 
simplex-wise filtrations
starting and ending with \emph{empty} complexes;
notice that any filtration can be converted into this form
by expanding the inclusions and attaching complexes to both ends.

Now let $\Fcal$ in Equation~(\ref{eqn:prelim-filt})
be a simplex-wise filtration starting and ending with 
empty complexes.
Then, each map $\Hm_*(K_i)\leftrightarrow \Hm_*(K_{i+1})$
in $\Hm_*(\Fcal)$ is either (i) injective with a one-dimensional cokernel
or (ii) surjective with a one-dimensional kernel.
The inclusion $K_i\leftrightarrow K_{i+1}$ provides
a \emph{birth} index $i+1$ (start of a persistence interval)
if 
$\Hm_*(K_i)\rightarrow \Hm_*(K_{i+1})$ is forward and injective,
or 
$\Hm_*(K_i)\leftarrow \Hm_*(K_{i+1})$ is backward and surjective.
Symmetrically, the inclusion provides
a \emph{death} index $i$ (end of a persistence interval)
if $\Hm_*(K_i)\rightarrow \Hm_*(K_{i+1})$ is forward and surjective,
or $\Hm_*(K_i)\leftarrow \Hm_*(K_{i+1})$ is backward and injective.
We denote the set of birth indices of $\Fcal$ as $\pinds(\Fcal)$
and the set of death indices of $\Fcal$ as $\ninds(\Fcal)$.

\section{Overview of main results}
\label{sec:overview}

In this section,
we detail all the update operations
with an overview of the main results for their computation.
The eight update operations (see Table~\ref{tab:complexities}) can be grouped into three types, i.e.,
switches, expansions, and contractions.
A switch is an interchange of two consecutive additions or deletions;
an expansion is an insertion of the addition and deletion of a simplex
in the middle;
a contraction is the reverse of an expansion.
Time complexities of the update algorithms for these operations
based on the two different approaches 
are listed in Table~\ref{tab:complexities}. 
In the table,
we denote the approach based on converting a zigzag filtration into
a non-zigzag one as \texttt{FZZ}-based (described in Section~\ref{sec:fzz-up}),
and the approach based on maintaining full representatives
for the intervals as \texttt{Rep}-based (described in Section~\ref{sec:update-alg}).
For each update operation,
let the filtration before and after the update be denoted as $\Fcal$
and $\Fcal'$ respectively,
which are both simplex-wise filtrations
starting and ending with empty complexes.
Then,
$\filtcnt$ in Table~\ref{tab:complexities} is the max length of $\Fcal$ and $\Fcal'$,
and $\simpcnt$ 
is the number of simplices in the \emph{total complex} $K$,
which is the union of all complexes in $\Fcal$ and $\Fcal'$
(notice that $\simpcnt\leq\filtcnt$).
As mentioned,
\texttt{FZZ}-based approaches cannot be directly applied to 
outward expansion and contraction due to change of adjacency relations
on the \emph{$\DG$-complex cells} (see Section~\ref{sec:adj-change}).
Hence, time complexities of \texttt{FZZ}-based approaches for these
two operations are left blank in Table~\ref{tab:complexities}.
In Section~\ref{sec:dpc-timing},
we present experimental results on computing vines and vineyards for dynamic point clouds
using the \texttt{Rep}-based update algorithms.

\begin{table}[htbp!]
    \centering
    \caption{Time complexities of update algorithms based on the two different approaches}
    \label{tab:complexities}
    \begin{tabular}{l cccc}
         \toprule
          & forward switch & backward switch & outward switch & inward switch \\
         \cmidrule{1-5}
         \texttt{FZZ}-based & $O(\filtcnt)$ & $O(\filtcnt)$ & $O(1)$ & $O(1)$ \\ 
         \texttt{Rep}-based & $O(\filtcnt\simpcnt)$ & $O(\filtcnt\simpcnt)$ & $O(\simpcnt^2+\filtcnt)$ & $O(\filtcnt)$ \\ 
         \midrule
         & outward expansion & outward contraction & inward expansion & inward contraction \\
         \cmidrule{1-5}
        \texttt{FZZ}-based & -- & -- & $O(\filtcnt^2)$ & $O(\filtcnt^2)$ \\
        \texttt{Rep}-based & $O(\filtcnt\simpcnt^2)$ & $O(\filtcnt\simpcnt^2)$ & $O(\filtcnt\simpcnt^2)$ & $O(\filtcnt\simpcnt^2)$ \\
         \bottomrule
    \end{tabular}
\end{table}

Notice that theorems describing the interval mapping for some operations
in this section have already been given in previous works.
Specifically,
Maria and Oudot~\cite{maria2014zigzag,maria2016computing} presented
a theorem on the forward/backward switches
(Transposition Diamond Principle~\cite{maria2014zigzag}).
Carlsson and de Silva~\cite{carlsson2010zigzag} presented
a theorem on the inward/outward switches
(Mayer-Vietoris Diamond Principle~\cite{carlsson2010zigzag,carlsson2019parametrized,carlsson2009zigzag-realvalue}).
Maria and Oudot~\cite{maria2014zigzag,maria2016computing} presented
a theorem on the inward expansion
(Injective/Surjective Diamond Principle~\cite{maria2014zigzag}). However, it was not
clear how the mappings given by these theorems can be computed with
efficient algorithms. We provide such algorithms in this paper.

Notice that outward expansion and inward/outward contractions
have not been considered elsewhere before, and our algorithms 
in Section~\ref{sec:out-expan}, \ref{sec:out-contra}, and Appendix~\ref{sec:in-contra}
implicitly provide theorems on their interval mappings.

\subsection{Update operations}
\label{sec:update-oper}

We now present all the update operations.
At the end of the subsection, we also provide a universality property
saying that every two zigzag filtrations
can be connected by a sequence of the update operations.

\medskip
\noindent
\textbf{{Forward switch}}
requires $\sG\nsubseteq\tG$:

\vspace{-1em}
\begin{center}
\begin{tikzpicture}
\tikzstyle{every node}=[minimum width=24em]
\node (a) at (0,0) {$\Fcal:K_0 \leftrightarrow\cdots\leftrightarrow K_{i-1}\inctosp{\sG}K_i\inctosp{\tG}K_{i+1}\leftrightarrow\cdots\leftrightarrow K_\filtcnt$}; 
\node (b) at (0,-0.6){$\Fcal':K_0\leftrightarrow\cdots\leftrightarrow K_{i-1}\inctosp{\tG} K'_i\inctosp{\sG} K_{i+1}\leftrightarrow\cdots\leftrightarrow K_\filtcnt$};
\path[->] (a.0) edge [bend left=90,looseness=1.5,arrows={-latex},dashed] (b.0);
\end{tikzpicture}
\end{center}

\vspace{-1em}\noindent
Notice that if $\sG\subseteq\tG$, then adding $\tG$ to $K_{i-1}$ in $\Fcal'$
does not produce a simplicial complex.

\medskip
\noindent
\textbf{{Backward switch}}
is the symmetric version of forward switch,
requiring $\tG\not\subseteq\sG$:

\vspace{-1em}
\begin{center}
\begin{tikzpicture}
\tikzstyle{every node}=[minimum width=24em]
\node (a) at (0,0) {$\Fcal: K_0 \leftrightarrow
\cdots
\leftrightarrow 
K_{i-1}\bakinctosp{\sG} 
K_i 
\bakinctosp{\tG} K_{i+1}
\leftrightarrow
\cdots \leftrightarrow K_\filtcnt$}; 
\node (b) at (0,-0.6){$\Fcal': K_0 \leftrightarrow
\cdots
\leftrightarrow 
K_{i-1}\bakinctosp{\tG} 
K'_i 
\bakinctosp{\sG} K_{i+1}
\leftrightarrow
\cdots \leftrightarrow K_\filtcnt$};
\path[->] (a.0) edge [bend left=90,looseness=1.5,arrows={-latex},dashed] (b.0);
\end{tikzpicture}
\end{center}

\vspace{-1.5em}\noindent

\medskip
\noindent
\textbf{{Outward switch}}
requires $\sG\neq\tG$:

\vspace{-1em}
\begin{center}
\begin{tikzpicture}
\tikzstyle{every node}=[minimum width=24em]
\node (a) at (0,0) {$\Fcal: K_0 \leftrightarrow
\cdots
\leftrightarrow 
K_{i-1}\inctosp{\sG} 
K_i 
\bakinctosp{\tG} K_{i+1}
\leftrightarrow
\cdots \leftrightarrow K_\filtcnt$}; 
\node (b) at (0,-0.6){$\Fcal': K_0 \leftrightarrow
\cdots
\leftrightarrow 
K_{i-1}\bakinctosp{\tG} 
K'_i 
\inctosp{\sG} K_{i+1}
\leftrightarrow
\cdots \leftrightarrow K_\filtcnt$};
\path[->] (a.0) edge [bend left=90,looseness=1.5,arrows={-latex},dashed] (b.0);
\end{tikzpicture}
\end{center}

\vspace{-1em}\noindent
Notice that if $\sG=\tG$, then we cannot delete $\tG$ from $K_{i-1}$ in $\Fcal'$
because $\tG\not\in K_{i-1}$.

\medskip
\noindent
\textbf{{Inward switch}}
is the reverse of outward switch,
requiring $\sG\neq\tG$:

\vspace{-1em}
\begin{center}
\begin{tikzpicture}
\tikzstyle{every node}=[minimum width=24em]
\node (a) at (0,0) {$\Fcal: K_0 \leftrightarrow
\cdots
\leftrightarrow 
K_{i-1}\bakinctosp{\sG} 
K_i 
\inctosp{\tG} K_{i+1}
\leftrightarrow
\cdots \leftrightarrow K_\filtcnt$}; 
\node (b) at (0,-0.6){$\Fcal': K_0 \leftrightarrow
\cdots
\leftrightarrow 
K_{i-1}\inctosp{\tG} 
K'_i 
\bakinctosp{\sG} K_{i+1}
\leftrightarrow
\cdots \leftrightarrow K_\filtcnt$};
\path[->] (a.0) edge [bend left=90,looseness=1.5,arrows={-latex},dashed] (b.0);
\end{tikzpicture}
\end{center}

\vspace{-1.5em}\noindent

\medskip
\noindent
\textbf{{Outward expansion}}
requires $\sG$ to be a simplex in $K_i$ without cofaces and $K'_{i-1}=K_{i}=K'_{i+1}$:

\vspace{-0.5em}
\begin{center}
\begin{tikzpicture}
\node (a) at (0,0) {$\Fcal: K_0 \leftrightarrow
\cdots\leftrightarrow K_{i-2}
\leftrightarrow 
K_i 
\leftrightarrow K_{i+2}\leftrightarrow
\cdots \leftrightarrow K_\filtcnt$}; 
\node (b) at (0,-0.6){$\Fcal': K_0 \leftrightarrow
\cdots\leftrightarrow K_{i-2}
\leftrightarrow 
K'_{i-1}
\bakinctosp{\sG} 
K'_i 
\inctosp{\sG} 
K'_{i+1}
\leftrightarrow K_{i+2}\leftrightarrow
\cdots \leftrightarrow K_\filtcnt$};
\draw[->,arrows={-latex},dashed] (a.0) .. controls (+6.5,+0) and (+7,-0.35) .. (b.0);
\end{tikzpicture}
\end{center}

\vspace{-1em}\noindent
To clearly show 
the correspondence of complexes in $\Fcal$ and $\Fcal'$, 
indices for $\Fcal$ are made nonconsecutive
in which $i-1$ and $i+1$ are skipped.

\medskip
\noindent
\textbf{{Outward contraction}}
is the reverse of outward expansion, requiring $K'_{i}=K_{i-1}=K_{i+1}$:

\vspace{-1em}
\begin{center}
\begin{tikzpicture}
\node (a) at (0,0) {$\Fcal: K_0 \leftrightarrow
\cdots\leftrightarrow K_{i-2}
\leftrightarrow 
K_{i-1}\bakinctosp{\sG} 
K_i 
\inctosp{\sG} K_{i+1}
\leftrightarrow K_{i+2}\leftrightarrow
\cdots \leftrightarrow K_\filtcnt$}; 
\node (b) at (0,-0.6){$\Fcal': K_0 \leftrightarrow
\cdots\leftrightarrow K_{i-2}
\leftrightarrow 
K'_{i}
\leftrightarrow K_{i+2}\leftrightarrow
\cdots \leftrightarrow K_\filtcnt$};
\draw[->,arrows={-latex},dashed] (a.0) .. controls (+7,-0.2) and (+6.5,-0.6) .. (b.0);
\end{tikzpicture}
\end{center}

\vspace{-1em}\noindent

\noindent
\textbf{{Inward expansion}}
is similar to outward expansion with the difference
that the two inserted arrows now pointing toward each other; 
it requires that $\sG\not\in K_i$, boundary simplices of $\sG$ be in $K_i$, and $K'_{i-1}=K_{i}=K'_{i+1}$:

\vspace{-0.5em}
\begin{center}
\begin{tikzpicture}
\node (a) at (0,0) {$\Fcal: K_0 \leftrightarrow
\cdots\leftrightarrow K_{i-2}
\leftrightarrow 
K_i 
\leftrightarrow K_{i+2}\leftrightarrow
\cdots \leftrightarrow K_\filtcnt$}; 
\node (b) at (0,-0.6){$\Fcal': K_0 \leftrightarrow
\cdots\leftrightarrow K_{i-2}
\leftrightarrow 
K'_{i-1}
\inctosp{\sG} 
K'_i 
\bakinctosp{\sG} 
K'_{i+1}
\leftrightarrow K_{i+2}\leftrightarrow
\cdots \leftrightarrow K_\filtcnt$};
\draw[->,arrows={-latex},dashed] (a.0) .. controls (+6.5,+0) and (+7,-0.35) .. (b.0);
\end{tikzpicture}
\end{center}

\noindent
\textbf{{Inward contraction}}
is the reverse of inward expansion, requiring $K'_{i}=K_{i-1}=K_{i+1}$:

\vspace{-1em}
\begin{center}
\begin{tikzpicture}
\node (a) at (0,0) {$\Fcal: K_0 \leftrightarrow
\cdots\leftrightarrow K_{i-2}
\leftrightarrow 
K_{i-1}\inctosp{\sG} 
K_i 
\bakinctosp{\sG} K_{i+1}
\leftrightarrow K_{i+2}\leftrightarrow
\cdots \leftrightarrow K_\filtcnt$}; 
\node (b) at (0,-0.6){$\Fcal': K_0 \leftrightarrow
\cdots\leftrightarrow K_{i-2}
\leftrightarrow 
K'_{i}
\leftrightarrow K_{i+2}\leftrightarrow
\cdots \leftrightarrow K_\filtcnt$};
\draw[->,arrows={-latex},dashed] (a.0) .. controls (+7,-0.2) and (+6.5,-0.6) .. (b.0);
\end{tikzpicture}
\end{center}

\paragraph{Universality of the operations.}
We present the following fact (proof in Appendix~\ref{sec:pf-prop-universality}):

\begin{proposition}\label{prop:universality}
Let $\Fcal_1,\Fcal_2$ be any two simplex-wise zigzag filtrations
starting and ending with empty complexes.
Then $\Fcal_1$ can be transformed into $\Fcal_2$ by
a sequence of the update operations listed above.
\end{proposition}

\subsection{Timing results for an example application}
\label{sec:dpc-timing}

We implement the representative-based update algorithms 
to compute vines and vineyards for dynamic point cloud (henceforth shortened as \emph{DPC})
as described in Appendix~\ref{sec:zzup-app}.
The source code is made public via: 
{\small\url{https://github.com/taohou01/zzup}}.
To demonstrate the efficiency of the representative-based update algorithms,
we compare their running time with that incurred by 
invoking a zigzag persistence algorithm
from scratch on each filtration.
For computing zigzag persistence from scratch,
we use an implementation\footnote{\url{https://github.com/taohou01/fzz}}
of the \textsc{FasZigzag} algorithm~\cite{dey2022fast},
which, according to the experiments in~\cite{dey2022fast},
gives the best running time for all
inputs among the algorithms tested.
For generating DPC datasets,
we use an implementation\footnote{\url{https://github.com/Nikorasu/PyNBoids}}
of the \emph{Boids}~\cite{kim2020spatiotemporal,reynolds1987flocks} model, 
which simulates the flocking behaviour of animals/objects such as birds.
As listed in Table~\ref{tab:timing}, two DPC datasets are generated, one with 10 boids (B)
moving over 100 time units (TU) and another with 15 boids moving over 20 time units.
For the Rips complexes changing over distance and time,
we only consider simplices up to dimension 3.
Table~\ref{tab:timing} also lists the numbers of different operations performed
for the datasets and the maximum length (MLen) of all filtrations generated.
From the table, we see that the accumulated computation time taken by our update algorithms
($\text{T}_\text{UP}$) 
is significantly less than that taken by invoking \textsc{FastZigzag} from scratch each time 
($\text{T}_\text{FS}$). 

\begin{table}[h!]
    \centering
    \caption{Test results on the two generated DPC datasets}
    \label{tab:timing}
    \begin{tabular}{rrrrrrrrrrr}
         \toprule
          B & TU & {fw\_sw} & {bw\_sw} & {ow\_sw} & {iw\_con} & {ow\_exp} & MLen & $\text{T}_\text{UP}$ & $\text{T}_\text{FS}$ \\
         \midrule
          10 & 100 & 23230 & 23000 & 42809 & 1646 & 1271 & 1200 & 0.54s & 35.19s \\
          15 & 20 & 736675 & 1107417 & 3284767 & 11093 & 4918 & 13732 & 11m1s & >28h\tablefootnote{The program ran for more than 28 hours and did not finish.} \\
         \bottomrule
    \end{tabular}
\end{table}

\section{Update algorithms based on \textsc{FastZigzag}}
\label{sec:fzz-up}

In this section,
we provide algorithms for the update operations
based on the \textsc{FastZigzag} algorithm~\cite{dey2022fast}.
We first briefly overview \textsc{FastZigzag}
(see~\cite{dey2022fast} for a detailed presentation),
and then describe how we utilize the algorithm for updates
with the help of the {transposition} operation proposed by Cohen-Steiner et al.~\cite{cohen2006vines}.
In Section~\ref{sec:adj-change},
we provide evidence for why the update algorithm for transpositions~\cite{cohen2006vines} cannot be
applied on outward expansion and contraction.

\subsection{Overview of \textsc{FastZigzag} algorithm}
\label{sec:fzz}

The \textsc{FastZigzag} algorithm
builds filtrations 
on the so-called \emph{$\DG$-complexes}~\cite{hatcher2002algebraic}.
Building blocks of $\DG$-complexes,
called \emph{cells} or \emph{$\DG$-cells},
are combinatorial equivalents
of simplices
(each $p$-cell is
formed by $p+1$ vertices
and has $p+1$ number of $(p-1)$-cells in the boundary)
whose common faces have more
relaxed forms~\cite{dey2022fast}.
Assuming a \emph{simplex}-wise zigzag filtration 
\[\Fcal:
\emptyset=
K_0\leftrightarrowsp{\fsimp{}{0}} K_1\leftrightarrowsp{\fsimp{}{1}}
\cdots 
\leftrightarrowsp{\fsimp{}{\filtcnt-1}} K_\filtcnt
=\emptyset
\]
consisting of simplicial complexes as input,
the \textsc{FastZigzag} algorithm
converts $\Fcal$ into the following \emph{cell}-wise non-zigzag filtration
consisting of $\DG$-complexes:
\[\ef:\emptyset\inctosp{\oG}\hat{K}_0\inctosp{\hatfsimp_0} 
\hat{K}_1\inctosp{\hatfsimp_{1}}
\cdots\inctosp{\hatfsimp_{\simpcnt-1}} 
\hat{K}_{\simpcnt}
\inctosp{\hatfsimp_{\simpcnt}}
\hat{K}_{\simpcnt+1}
\inctosp{\hatfsimp_{\simpcnt+1}}
\cdots
\inctosp{\hatfsimp_{\filtcnt-1}} \hat{K}_{\filtcnt},\]
where $\filtcnt=2\simpcnt$
($\filtcnt$ is an even number because an added simplex
must be eventually deleted in ${\Fcal}$).
In $\ef$, $\oG$ is a vertex used for \emph{coning}.
Cells $\hatfsimp_0,\hatfsimp_1,\ldots,\hatfsimp_{\simpcnt-1}$ 
are copies of all added simplices in $\Fcal$ 
with the order of addition preserved.
Cells 
$\hatfsimp_{\simpcnt},\hatfsimp_{\simpcnt+1},\ldots,\hatfsimp_{\filtcnt-1}$ 
are \emph{cones} of those $\DG$-cells corresponding to all simplices deleted in $\Fcal$,
with the order reversed.

\begin{definition}
In $\Fcal$ or $\ef$, 
let each addition or deletion of a simplex
be uniquely identified by its index in the filtration,
e.g., the index of $K_i\leftrightarrowsp{\fsimp{}{i}}K_{i+1}$ in $\Fcal$ is $i$.
Then, 
the \defemph{creator} of an interval $[b,d]\in\Pers_*(\Fcal)\text{ or }\Pers_*(\ef)$
is an addition/deletion indexed at $b-1$,
and the \defemph{destroyer} of the interval
is an addition/deletion indexed at $d$.\footnote{We notice the following 
(whose reasons are evident from later contents): 
(i) the index for
the initial addition of $\oG$ in $\ef$ is not needed and therefore is undefined; 
(ii) we require $[b,d]\in\Pers_*(\ef)$ to be a finite interval ($d<\filtcnt$).}
\end{definition}
Notice that  creators and destroyers defined above are the same 
as the `simplex pairs' in standard persistence~\cite{edelsbrunner2000topological}.

As stated previously, 
each $\hatfsimp_i$ in $\ef$ for $0\leq i<\simpcnt$ corresponds
to an addition in $\Fcal$, and
each $\hatfsimp_i$ for $\simpcnt\leq i<\filtcnt$ corresponds
to a deletion in $\Fcal$.
This naturally defines a bijection $\Phi$
from the additions and deletions in $\Fcal$ to the additions in $\ef$ excluding $\oG$.
Moreover, for simplicity, we let the domain and codomain of $\Phi$
be the sets of indices for the additions and deletions.
We then summarize the interval mapping~\cite{dey2022fast} for \textsc{FastZigzag}
as follows:

\begin{theorem}
\label{thm:fzz-intv-map}
Given $\Pers_*(\ef)$, one can retrieve $\Pers_*(\Fcal)$ 
using the following bijective mapping
from the set of finite intervals of $\Pers_*(\ef)$
to $\Pers_*(\Fcal)$:
an 
interval $[b,d]\in\Pers_*(\ef)$ 
with a creator indexed at $b-1$ and a destroyer indexed at $d$
is mapped to an interval $I\in \Pers_*(\Fcal)$ with
the same creator and destroyer indexed at 
$\Phi\inv(b-1)$ and $\Phi\inv(d)$ respectively.
Specifically, if $\Phi\inv(b-1)<\Phi\inv(d)$, then $I=[\Phi\inv(b-1)+1,\Phi\inv(d)]$,
where $\Phi\inv(b-1)$ indexes the creator and $\Phi\inv(d)$ indexes the destroyer;
otherwise,
$I=[\Phi\inv(d)+1,\Phi\inv(b-1)]$,
where $\Phi\inv(d)$ indexes the creator and $\Phi\inv(b-1)$ indexes the destroyer.
\end{theorem}
Notice that the only infinite interval in $\Pers_*(\ef)$ is $[0,\infty)$ in dimension 0 created by adding $\oG$~\cite{cohen2009extending}.
Also, the dimension of an interval needs to be properly set when mapping the intervals as described above.
For brevity, we omit the details; see Proposition~15 and~19 in~\cite{dey2022fast}.

\subsection{Using \textsc{FastZigzag} for updates}

We utilize the conversion of a zigzag filtration into a non-zigzag filtration 
in \textsc{FastZigzag}
and the transposition operation proposed in~\cite{cohen2006vines} 
to update the barcodes for the six operations in Table~\ref{tab:complexities}.
Notice that a transposition is indeed a forward switch applied to a non-zigzag
filtration, which can be updated in linear time w.r.t the filtration's length~\cite{cohen2006vines}.
For the update, we maintain the following core data structures
for the zigzag filtration and its converted non-zigzag filtration before and after the operation:
(i) two arrays $\phi,\phi'$ of size $\filtcnt$ encoding the mapping $\Phi$ and $\Phi\inv$;
(ii) another array $\Pi$ of size $\filtcnt$ recording the pairing of creators and destroyers
for the converted non-zigzag filtration.
By Theorem~\ref{thm:fzz-intv-map},
the creator-destroyer pairing for the original zigzag filtration 
can be derived from the core data structures and hence the barcode can be easily updated.

\paragraph{Outward/inward switch.}
Since the corresponding non-zigzag filtration before and after outward/inward switch
stays the same, we only need to update entries in $\phi,\phi'$ that change. So the time complexity is $O(1)$.

\paragraph{Forward/backward switch.}
Corresponding to a forward/backward switch on the original zigzag filtration,
there is a transposition of two additions in the converted non-zigzag filtration.
Updating the pairing in $\Pi$ then takes $O(\filtcnt)$ time using the algorithm in~\cite{cohen2006vines}.
Notice that $\phi,\phi'$ stay the same before and after the switch.
So forward/backward switch takes $O(\filtcnt)$ time.

\paragraph{Inward expansion.}
For the following inward expansion:

\vspace{-0.5em}
\begin{center}
\begin{tikzpicture}
\node (a) at (0,0) {$\Fcal: K_0 \leftrightarrow
\cdots\leftrightarrow K_{i-2}
\leftrightarrow 
K_i 
\leftrightarrow K_{i+2}\leftrightarrow
\cdots \leftrightarrow K_\filtcnt$}; 
\node (b) at (0,-0.6){$\Fcal': K_0 \leftrightarrow
\cdots\leftrightarrow K_{i-2}
\leftrightarrow 
K'_{i-1}
\inctosp{\sG} 
K'_i 
\bakinctosp{\sG} 
K'_{i+1}
\leftrightarrow K_{i+2}\leftrightarrow
\cdots \leftrightarrow K_\filtcnt$};
\draw[->,arrows={-latex},dashed] (a.0) .. controls (+6.5,+0) and (+7,-0.35) .. (b.0);
\end{tikzpicture}
\end{center}

\vspace{-1em}\noindent
we first attach the additions of $\hat{\sG}$ and $\oG\cdot\hat{\sG}$ to the end of the non-zigzag filtration
corresponding to $\Fcal$,
where $\hat{\sG}$ is the $\DG$-cell corresponding to the inserted simplex $\sG$.
Attaching the two additions needs performing two rounds of reductions in 
the persistence algorithm~\cite{cohen2006vines,edelsbrunner2000topological}
and therefore takes $O(\filtcnt^2)$ time.
We then perform transpositions (and update $\Pi$ accordingly) 
to switch the additions of $\hat{\sG}$ and $\oG\cdot\hat{\sG}$
to proper positions so that the non-zigzag filtration
correctly corresponds to the new zigzag filtration $\Fcal'$.
After the transpositions, we also perform necessary updates for $\phi,\phi'$
which takes $O(\filtcnt)$ time.
Since $O(\filtcnt)$ transpositions are performed,
the time complexity of inward expansion
is $O(\filtcnt^2)$.

\paragraph{Inward contraction.}
The algorithm for inward contraction follows the reverse process of that for inward expansion:
we first bring the additions of $\hat{\sG}$ and $\oG\cdot\hat{\sG}$
(defined similarly as previous) to the end of the non-zigzag filtration
by transpositions, and then delete the two additions at the end.
Since $O(\filtcnt)$ transpositions are performed
and updating $\phi,\phi'$ takes $O(\filtcnt)$ time,
the time complexity of inward contraction
is $O(\filtcnt^2)$.

\subsection{Change of adjacency in outward expansion/contraction}
\label{sec:adj-change}

We now explain why the update algorithm for transposition in~\cite{cohen2006vines} cannot be 
applied directly on outward expansion and contraction. 
Consider the following outward expansion:

\vspace{-0.5em}
\begin{center}
\begin{tikzpicture}
\node (a) at (0,0) {$\Fcal: K_0 \leftrightarrow
\cdots\leftrightarrow K_{i-2}
\leftrightarrow 
K_i 
\leftrightarrow K_{i+2}\leftrightarrow
\cdots \leftrightarrow K_\filtcnt$}; 
\node (b) at (0,-0.6){$\Fcal': K_0 \leftrightarrow
\cdots\leftrightarrow K_{i-2}
\leftrightarrow 
K'_{i-1}
\bakinctosp{\sG} 
K'_i 
\inctosp{\sG} 
K'_{i+1}
\leftrightarrow K_{i+2}\leftrightarrow
\cdots \leftrightarrow K_\filtcnt$};
\draw[->,arrows={-latex},dashed] (a.0) .. controls (+6.5,+0) and (+7,-0.35) .. (b.0);
\end{tikzpicture}
\end{center}

\vspace{-1em}\noindent
where $\sG$ is a $\Dim$-simplex.
Let $\ef$ and $\ef'$ be the non-zigzag filtrations constructed 
by \textsc{FastZigzag} for $\Fcal$ and $\Fcal'$ respectively.
Since $\sG\in K_i$, $\sG$ must have been added in $\Fcal$ before
$K_i$. Let $\hat{\sG}^0$ be the $\DG$-cell in $\ef$ corresponding to the most 
recent addition of $\sG$ before $K_i$ in $\Fcal$.
Furthermore, if there are cells added after $\hat{\sG}^0$ 
in $\ef$ which are copies of $\sG$,
let $\hat{\sG}^1$ be the first such cell;
otherwise, let $\hat{\sG}^1$ be the first coned cell added in $\ef$.
Then, let $\GG$ be the set of $(\Dim+1)$-cells added between $\hat{\sG}^0$ and $\hat{\sG}^1$
in $\ef$ whose corresponding simplices in $\Fcal$ 
contain $\sG$ in boundaries.
By the construction of $\ef$~\cite{dey2022fast},
cells in $\GG$ must have $\hat{\sG}^0$ as a boundary $\Dim$-cell.

Now consider $\ef'$. Due to its construction,
we must have that there is a $\Dim$-cell $\hat{\sG}'$ in $\ef'$
corresponding to the addition $K'_i\inctosp{\sG}K'_{i+1}$.
Notice that $\hat{\sG}'$ does not appear in $\ef$ and must be
added between $\hat{\sG}^0$ and $\hat{\sG}^1$ in $\ef'$
(because $\hat{\sG}^0$ corresponds to the most 
recent addition of $\sG$ before $K_i$).
Let $\GG'$ be the set of $(\Dim+1)$-cells added between $\hat{\sG}'$ and $\hat{\sG}^1$
in $\ef'$ whose corresponding simplices in $\Fcal'$
contain $\sG$ in boundaries.
Then, cells in $\GG'$ must now have $\hat{\sG}'$ as a boundary $\Dim$-cell~\cite{dey2022fast}.
Notice that $\GG'\subseteq\GG$.
Therefore, for the $(\Dim+1)$-cells in $\GG'$,
one boundary $\Dim$-cell
$\hat{\sG}^0$ changes to $\hat{\sG}'$
when going from $\ef$ to $\ef'$.
However, the update algorithm for transposition in~\cite{cohen2006vines}
cannot change the boundary (adjacency) relation for cells,
even if we have switched $\hat{\sG}'$ to the correct position by transpositions.
Notice that we also need to add a coned $\DG$-cell corresponding to
the deletion $K'_{i-1}\bakinctosp{\sG}K'_i$ in $\ef'$;
the change in $\ef'$ by adding this coned cell is similar as above
and details are omitted.

Since outward contraction is the reverse of outward expansion,
the change in the converted non-zigzag filtration is symmetric to previous:
one boundary $\Dim$-cell in some $(\Dim+1)$-cells changes to an 
earlier copy of a $\Dim$-simplex $\sG$.
Hence, we also cannot directly apply the update algorithm for transposition~\cite{cohen2006vines}
on outward contraction.

\section{Update algorithms based on maintaining full representatives}
\label{sec:update-alg}

In this section, we present update algorithms for {all} the eight operations
based on explicit maintenance of representatives for the persistence intervals.
As stated earlier, the representative-based algorithms are useful 
in the following situations:
(i) an application that requires outward expansion and contraction which
cannot use the \textsc{FastZigzag}-based approach
(see Section~\ref{sec:adj-change});
(ii) an application that requires explicit updates of representatives (see Appendix~\ref{sec:zzup-app}).

We first present the update algorithms
for outward expansion and contraction 
and then present the algorithms for the remaining operations.
Due to space restrictions,
algorithms for some operations are put into Appendix~\ref{sec:update-alg-cont}.
We begin by
laying some foundations 
for the update algorithms
in Section~\ref{sec:principles},
where we formally present the definition of representatives 
(adapted from~\cite{maria2014zigzag}; see Definition~\ref{dfn:rep-seq}).
Notations for all operations adopted in Section~\ref{sec:update-oper} are retained, 
e.g.,
$\Fcal$ and $\Fcal'$
denote the filtration before and after the update respectively.
Before the update,
we assume that we are given the barcode $\Pers_*(\Fcal)$
and the representatives for their intervals.
Our goal is to compute $\Pers_*(\Fcal')$ and the representatives for
$\Pers_*(\Fcal')$ based on what is given.
This is achieved by adjusting the pairing of birth and death indices 
for $\Fcal'$ so that we can identify representatives
for every interval induced from the pairing.
Proposition~\ref{prop:pn-paring-w-rep} in Section~\ref{sec:principles}
justifies such an approach.
Hence, the correctness of the algorithms in this section follows
from the correctness of the representatives being updated,
which is implicit in our description.

\subsection{Principles of representative-based updates}
\label{sec:principles}

We first present the following proposition useful to many of the update algorithms:

\begin{proposition}\label{prop:homolog-cyc-no-new-simp}
For a simplex-wise inclusion $X\inctosp{\sG}X'$
of two complexes,
let $z$ be a cycle in $X'$ homologous to 
a cycle in $X$.
Then, $\sG\not\in z$.
\end{proposition}
\begin{proof}
Let $y$ be the cycle in $X$ that $z$ is homologous to.
We have $z=y+\partial(A)$ for $A\subseteq X'$.
Since $\sG\not\in y$ and $\sG\not\in\partial(A)$
($\sG$ has no cofaces in $X'$), we have that $\sG\not\in z$.
\end{proof}

Throughout the subsection, let $\Fcal: \emptyset=K_0 \leftrightarrowsp{\fsimp_0} K_1 \leftrightarrowsp{\fsimp_1}
\cdots \leftrightarrowsp{\fsimp_{\filtcnt-1}} K_\filtcnt=\emptyset$ be a simplex-wise filtration
starting and ending with empty complexes.

\begin{definition}[Representative]
\label{dfn:rep-seq}
Let $[b,d]\subseteq\Set{1,\ldots,\filtcnt-1}$ be an interval.
A \defemph{$\Dim$-th representative sequence} 
$($also simply called \defemph{$\Dim$-th representative}$)$
for $[b,d]$ consists of a sequence of $\Dim$-cycles
$\Set{\cyc_i\in \Cyc_\Dim(K_i)\given b\leq i\leq d}$
and a sequence of $(\Dim+1)$-chains
$\Set{\chn_i\given b-1\leq i\leq d}$,
typically denoted as
\[\chn_{b-1}\dashleftarrow\cyc_{b}\rseqlrarr{\chn_{b}}\cdots
\rseqlrarr{\chn_{d-1}}\cyc_{d}\dashrightarrow\chn_{d},\]
such that for each $i$ with $b\leq i<d$:
\begin{itemize}
    \item if $K_i\incto K_{i+1}$ is forward, then $\chn_i\in\Chn_{\Dim+1}(K_{i+1})$
    and $\cyc_i+\cyc_{i+1}=\partial(\chn_i)$ in $K_{i+1}$;
    \item if $K_i\bakincto K_{i+1}$ is backward, then $\chn_i\in\Chn_{\Dim+1}(K_{i})$
    and $\cyc_i+\cyc_{i+1}=\partial(\chn_i)$ in $K_{i}$.
\end{itemize}
Furthermore, the sequence 
satisfies the additional conditions:
\begin{description}
    \item[Birth condition:] \label{itm:dfn-rep-birth}
    If $K_{b-1}\bakinctosp{\fsimp_{b-1}} K_{b}$ is backward,
    then $\cyc_b=\partial(\chn_{b-1})$ 
    for $\chn_{b-1}$ a $(\Dim+1)$-chain in $K_{b-1}$ containing
    $\fsimp_{b-1}$;
    if $K_{b-1}\inctosp{\fsimp_{b-1}} K_{b}$ is forward,
    then $\fsimp_{b-1}\in\cyc_b$
    and $\chn_{b-1}$ is undefined.
    
    \item[Death condition:] \label{itm:dfn-rep-death}
    If $K_{d}\inctosp{\fsimp_{d}} K_{d+1}$ is forward,
    then $\cyc_d=\partial(\chn_{d})$ 
    for $\chn_{d}$ a $(\Dim+1)$-chain in $K_{d+1}$ containing
    $\fsimp_{d}$;
    if $K_{d}\bakinctosp{\fsimp_{d}} K_{d+1}$ is backward,
    then $\fsimp_{d}\in\cyc_d$
    and $\chn_{d}$ is undefined.
\end{description}
Moreover, 
we relax the above definition and define a \defemph{post-birth} 
representative sequence for $[b,d]$ by ignoring the death condition.
Similarly, we define a \defemph{pre-death} 
representative sequence for $[b,d]$ by ignoring the birth condition.
\end{definition}
We sometimes ignore the undefined chains 
(e.g., $\chn_{b-1}$ or $\chn_{d}$) for $[b,d]$ when denoting
a representative sequence.
Also, the cycle $z_i$ in the above definition is
called the {\it representative $\Dim$-cycle at} index $i$
for $[b,d]$.

The following proposition from~\cite{dey2021graph} 
says that as long as 
one has a pairing of the birth and death indices s.t.\ each interval
induced by the pairing has a representative sequence,
one has the barcode.

\begin{proposition}\label{prop:pn-paring-w-rep}
Let $\pi:\pinds(\Fcal)\to\ninds(\Fcal)$ be a bijection.
If every $b\in\pinds(\Fcal)$ satisfies that $b\leq\pi(b)$ and the interval $[b,\pi(b)]$
has a representative sequence,
then $\Pers_*(\Fcal)=\Set{[b,\pi(b)]\given b\in\pinds(\Fcal)}$.
\end{proposition}

\begin{definition}[Birth/death order~\cite{maria2014zigzag}]
\label{dfn:bd-ordere}
Define a total order `$\bles$' for the \defemph{birth} indices in $\Fcal$.
For two indices $b_1,b_2\in \Set{1,\ldots,\filtcnt-1}$
s.t.\ $b_1\neq b_2$,
one has that
$b_1\bles b_2$
iff:
    (i)
    $b_1<b_2$ and 
    $K_{b_2-1}\incto K_{b_2}$ is forward,
    or (ii) $b_1>b_2$ and 
    $K_{b_1-1}\bakincto K_{b_1}$ is backward.
Symmetrically,
define a total order `$\dles$' for the \defemph{death} indices in $\Fcal$.
For two indices $d_1,d_2\in \Set{1,\ldots,\filtcnt-1}$
s.t.\ $d_1\neq d_2$,
one has that $d_1\dles d_2$
iff:
    (i)
    $d_1>d_2$ and 
    $K_{d_2}\bakincto K_{d_2+1}$ is backward,
    or (ii) $d_1<d_2$ and 
    $K_{d_1}\incto K_{d_1+1}$ is forward.
\end{definition}
The motivation behind the above orders is as follows:
for two intervals $[b_1,i],[b_2,i]$ s.t.\ $b_1\bles b_2$,
a post-birth representative for $[b_1,i]$ can always be
`added to' a post-birth representative for $[b_2,i]$ (see Section~\ref{sec:rep-oper}).
A similar fact holds for the order `$\dles$'.

\begin{definition}
Two non-disjoint intervals $[b_1,d_1],[b_2,d_2]\subseteq\Set{1,\ldots,\filtcnt-1}$ %
are called \defemph{comparable} if $b_1\bles b_2$ and $d_1\dles d_2$,
or $b_2\bles b_1$ and $d_2\dles d_1$.
Also, we use `$[b_1,d_1]\iles [b_2,d_2]$' to denote the situation
that $b_1\bles b_2$ and $d_1\dles d_2$.
\end{definition}

\subsubsection{Operations on representatives}
\label{sec:rep-oper}

We present some operations on representative sequences useful for 
the update algorithms.

\paragraph{Sum for post-birth representatives.}
For the following $\Dim$-th post-birth representatives 
\begin{align*}
\begin{split}
\rseq_1:\chn_{b_1-1}\dashleftarrow\cyc_{b_1}\rseqlrarr{\chn_{b_1}}\cdots
\rseqlrarr{\chn_{i-1}}\cyc_{i},\quad
\rseq_2:\chn'_{b_2-1}\dashleftarrow\cyc'_{b_2}\rseqlrarr{\chn'_{b_2}}\cdots
\rseqlrarr{\chn'_{i-1}}\cyc'_{i}
\end{split}
\end{align*}
for two intervals $[b_1,i],[b_2,i]\subseteq\Set{1,\ldots,\filtcnt-1}$
where $b_1\bles b_2$,
we define a \emph{sum} of $\rseq_1$ and $\rseq_2$, denoted $\rseq_1\brepsum\rseq_2$.
If $b_1<b_2$ 
(i.e., $K_{b_2-1}\incto K_{b_2}$ is forward), 
then $\rseq_1\brepsum\rseq_2$ is defined as:
\[\rseq_1\brepsum\rseq_2:\cyc_{b_2}+\cyc'_{b_2}
\rseqlrarr{\chn_{b_2}+\chn'_{b_2}}\cdots
\rseqlrarr{\chn_{i-1}+\chn'_{i-1}}\cyc_{i}+\cyc'_{i};\]
if $b_1>b_2$ 
(i.e., $K_{b_1-1}\bakincto K_{b_1}$ is backward), 
then $\rseq_1\brepsum\rseq_2$ is defined as:
\begin{align*}
\begin{split}
\rseq_1\brepsum\rseq_2:
\chn'_{b_2-1}\dashleftarrow\cyc'_{b_2}\rseqlrarr{\chn'_{b_2}}\cdots
\rseqlrarr{\chn'_{b_1-2}}\cyc'_{b_1-1}
\rseqlarr{\chn_{b_1-1}+\chn'_{b_1-1}}\cyc_{b_1}+\cyc'_{b_1}
\rseqlrarr{\chn_{b_1}+\chn'_{b_1}}&\\
\cdots
\rseqlrarr{\chn_{i-1}+\chn'_{i-1}}\cyc_{i}+\cyc'_{i}&.
\end{split}
\end{align*}
It can be verified that $\rseq_1\brepsum\rseq_2$ 
is a $\Dim$-th post-birth representative for $[b_2,i]$.
For example, when $b_1<b_2$, 
since
$\fsimp_{b_2-1}\not\in\cyc_{b_2}$ (Proposition~\ref{prop:homolog-cyc-no-new-simp})
and $\fsimp_{b_2-1}\in\cyc'_{b_2}$,
we have that $\fsimp_{b_2-1}\in\cyc_{b_2}+\cyc'_{b_2}$.

\paragraph{Sum for pre-death representatives.} Symmetrically, for 
$\Dim$-th pre-death representatives
\begin{align*}
\begin{split}
\rseq_1:\cyc_{i}\rseqlrarr{\chn_{i}}\cdots
\rseqlrarr{\chn_{d_1-1}}\cyc_{d_1}\dashrightarrow\chn_{d_1},\quad
\rseq_2:\cyc'_{i}\rseqlrarr{\chn'_{i}}\cdots
\rseqlrarr{\chn'_{d_2-1}}\cyc'_{d_2}\dashrightarrow\chn'_{d_2}
\end{split}
\end{align*}
for intervals $[i,d_1]$, $[i,d_2]$
s.t.\ $d_1\dles d_2$,
we define a \emph{sum} $\rseq_1\drepsum\rseq_2$
as a $\Dim$-th pre-death representative for $[i,d_2]$.
If $d_1>d_2$ 
(i.e., $K_{d_2}\bakincto K_{d_2+1}$ is backward), 
then $\rseq_1\drepsum\rseq_2$ is:
\[\rseq_1\drepsum\rseq_2:\cyc_{i}+\cyc'_{i}
\rseqlrarr{\chn_{i}+\chn'_{i}}\cdots
\rseqlrarr{\chn_{d_2-1}+\chn'_{d_2-1}}\cyc_{d_2}+\cyc'_{d_2};\]
if $d_1<d_2$ 
(i.e., $K_{d_1}\incto K_{d_1+1}$ is forward), 
then $\rseq_1\drepsum\rseq_2$ is:
\begin{align*}
\begin{split}
\rseq_1\drepsum\rseq_2:
\cyc_{i}+\cyc'_{i}\rseqlrarr{\chn_{i}+\chn'_{i}}\cdots
\rseqlrarr{\chn_{d_1-1}+\chn'_{d_1-1}}\cyc_{d_1}+\cyc'_{d_1}
\rseqrarr{\chn_{d_1}+\chn'_{d_1}}\cyc'_{d_1+1}
\rseqlrarr{\chn'_{d_1+1}}&\\
\cdots
\rseqlrarr{\chn'_{d_2-1}}\cyc'_{d_2}\dashrightarrow\chn'_{d_2}&.
\end{split}
\end{align*}

\paragraph{Concatenation.}
Let $\rseq_1$ be a $\Dim$-th  post-birth representative for $[b,i]$
and $\rseq_2$ be a $\Dim$-th  pre-death representative for $[i,d]$,
which are of the forms:
\[\rseq_1:\chn_{b-1}\dashleftarrow\cyc_{b}\rseqlrarr{\chn_{b}}\cdots
\rseqlrarr{\chn_{i-1}}\cyc_{i},\quad
\rseq_2:\cyc'_{i}\rseqlrarr{\chn'_{i}}\cdots
\rseqlrarr{\chn'_{d-1}}\cyc'_{d}\dashrightarrow\chn'_{d}.
\]
If $\cyc_{i}$ is homologous to $\cyc'_{i}$ in $K_i$,
i.e., $\cyc_i+\cyc'_i=\partial(A)$ for $A \in \Chn_{\Dim+1}(K_i)$,
then we define a \emph{concatenation} of $\rseq_1$ and $\rseq_2$,
denoted $\rseq_1\rconcat\rseq_2$,
as:
\[
\rseq_1\rconcat\rseq_2:\chn_{b-1}\dashleftarrow\cyc_{b}\rseqlrarr{\chn_{b}}\cdots
\rseqlrarr{\chn_{i-2}}\cyc_{i-1}
\rseqlrarr{\chn_{i-1}+A}\cyc'_{i}\rseqlrarr{\chn'_{i}}\cdots
\rseqlrarr{\chn'_{d-1}}\cyc'_{d}\dashrightarrow\chn'_{d}.
\]
Notice that $\rseq_1\rconcat\rseq_2$ is a 
$\Dim$-th representative sequence for $[b,d]$.

\paragraph{Prefix, suffix, and sum for representatives.}
Let 
\[\rseq:\chn_{b-1}\dashleftarrow\cyc_{b}\rseqlrarr{\chn_{b}}\cdots
\rseqlrarr{\chn_{d-1}}\cyc_{d}\dashrightarrow\chn_{d}\]
be a $\Dim$-th representative sequence for an interval 
$[b,d]\subseteq\Set{1,\ldots,\filtcnt-1}$. For an index $i\in[b,d]$,
define a {\it prefix} $\rseq\rsprefix{i}$ as
a $\Dim$-th post-birth representative for $[b,i]$:
\[\rseq\rsprefix{i}:\chn_{b-1}\dashleftarrow\cyc_{b}\rseqlrarr{\chn_{b}}\cdots
\rseqlrarr{\chn_{i-1}}\cyc_{i}.\]
Similarly, define a {\it suffix} $\rseq\rssuffix{i}$ as
a $\Dim$-th pre-death representative for $[i,d]$:
\[\rseq\rssuffix{i}:\cyc_{i}\rseqlrarr{\chn_{i}}\cdots
\rseqlrarr{\chn_{d-1}}\cyc_{d}\dashrightarrow\chn_{d}.\]

Let $[b_1,d_1],[b_2,d_2]\subseteq\Set{1,\ldots,\filtcnt-1}$ 
be two intervals containing a common index $i$
and let $\rseq_1,\rseq_2$
be $\Dim$-th representative sequences for $[b_1,d_1],[b_2,d_2]$
respectively.
We define a \emph{sum} of $\rseq_1$ and $\rseq_2$, denoted $\rseq_1\repsum\rseq_2$,
as a $\Dim$-th representative sequence for the interval 
$[\Max_{\bles}\Set{b_1,b_2},\Max_{\dles}\Set{d_1,d_2}]$:
\[\rseq_1\repsum\rseq_2:=\big(\rseq_1\rsprefix{i}\brepsum\rseq_2\rsprefix{i}\big)
\rconcat\big(\rseq_1\rssuffix{i}\drepsum\rseq_2\rssuffix{i}\big).\]
Notice that the values of $\rseq_1\repsum\rseq_2$ are indeed irrelevant to the choice of $i$.
Specifically, if $[b_1,d_1]\iles [b_2,d_2]$,
then $\rseq_1\repsum\rseq_2$ is a $\Dim$-th representative
for $[b_2,d_2]$.

\subsubsection{Data structures for representatives}

We use a simple data structure to implement
a $\Dim$-th representative sequence
for an interval $[b,d]$. Using an array,
each index $i\in[b,d]$ is associated with a pointer 
to the $\Dim$-cycle at $i$. 
Notice that consecutive indices in $[b,d]$ may be associated with
the same $\Dim$-cycle. In this case, to save memory space, 
we let the pointers for these indices point to the same 
memory location.
We also do the similar thing for the $(\Dim+1)$-chains.
Let $\simpcnt$ be the number of simplices in $\bigunion_{i=0}^\filtcnt K_i$.
Then, the summation of two representative sequences takes $O(\filtcnt\simpcnt)$ time
because $[b,d]$ has $O(\filtcnt)$ indices
and adding two cycles or chains at a index
takes $O(\simpcnt)$ time.

\subsection{Outward expansion}
\label{sec:out-expan}

Recall that an outward expansion is the following operation:

\vspace{-0.5em}
\begin{center}
\begin{tikzpicture}
\node (a) at (0,0) {$\Fcal: K_0 \leftrightarrow
\cdots\leftrightarrow K_{i-2}
\leftrightarrow 
K_i 
\leftrightarrow K_{i+2}\leftrightarrow
\cdots \leftrightarrow K_\filtcnt$}; 
\node (b) at (0,-0.6){$\Fcal': K_0 \leftrightarrow
\cdots\leftrightarrow K_{i-2}
\leftrightarrow 
K'_{i-1}
\bakinctosp{\sG} 
K'_i 
\inctosp{\sG} 
K'_{i+1}
\leftrightarrow K_{i+2}\leftrightarrow
\cdots \leftrightarrow K_\filtcnt$};
\draw[->,arrows={-latex},dashed] (a.0) .. controls (+6.5,+0) and (+7,-0.35) .. (b.0);
\end{tikzpicture}
\end{center}

\vspace{-1em}\noindent
where $K'_{i-1}=K_{i}=K'_{i+1}$.
We also assume that $\sG$ is a $\Dim$-simplex.
Notice that
indices for $\Fcal$ are nonconsecutive
in which $i-1$ and $i+1$ are skipped.

For the update,
we first determine whether the induced map 
$\Hm_*(K'_{i-1})\leftarrow \Hm_*(K'_i)$
is injective
or 
surjective
by checking whether 
$\sG$ is in a $\Dim$-cycle in $K'_{i-1}$
(injective)
or not
(surjective).
Let $\Set{I_j\given j\in\Bcal}$ be the set of 
intervals in $\Pers_\Dim(\Fcal)$ containing $i$, where $\Bcal$
is an indexing set.
Also, let $\cyc^j_i$ be the representative $\Dim$-cycle at index $i$
for $I_j$.
Note that $\Set{[\cyc^j_i]\given j\in\Bcal}$ is a basis for
$\Hm_\Dim(K_i)=\Hm_\Dim(K'_{i-1})$.
Then, we claim that (i) $\sG$ is in a $\Dim$-cycle in $K'_{i-1}$
$\Leftrightarrow$
(ii) $\sG\in\cyc^j_i$ for a $j\in\Bcal$.
Hence, to determine the injectivity/surjectivity,
we only need to check condition (ii).
To prove the claim, let $\cyc\subseteq K'_{i-1}$ be a $\Dim$-cycle containing $\sG$.
Then, $\cyc=\sum_{j\in\LG}\cyc^j_i+x$, where $\LG\subseteq\Bcal$
and $x$ is a $\Dim$-boundary in $K'_{i-1}$.
We have that $\sG\not\in x$ because $\sG$ has no cofaces in $K'_{i-1}$.
Hence, $\sG\in\sum_{j\in\LG}\cyc^j_i$, which implies condition (ii).
This proves the `only if' part of the claim,
and the proof for the `if' part is obvious.

\subsubsection{$\Hm_*(K'_{i-1})\leftarrow \Hm_*(K'_i)$ is surjective}
\label{sec:out-expan-surj}
The only difference of $\Pers_*(\Fcal)$ and $\Pers_*(\Fcal')$
in this case 
is that 
there is a new interval
$[i,i]$ in $\Pers_*(\Fcal')$
with the representative $\chn_{i-1}\dashleftarrow\cyc_{i}\dashrightarrow\chn_{i}$,
where $\cyc_i=\partial(\sG)$ and $\chn_{i-1}=\chn_{i}=\sG$.
Let $[b,d]$ be an interval in $\Pers_*(\Fcal)$.
If $i\not\in[b,d]$, 
the representative for $[b,d]\in\Pers_*(\Fcal)$
can be directly used as a representative for $[b,d]\in\Pers_*(\Fcal')$.
If $i\in[b,d]$, 
let 
\[\rseq:\chn_{b-1}\dashleftarrow\cyc_{b}\rseqlrarr{\chn_{b}}\cdots
\rseqlrarr{\chn_{i-3}}\cyc_{i-2}
\rseqlrarr{\chn_{i-2}}\cyc_{i}
\rseqlrarr{\chn_{i}}\cyc_{i+2}
\rseqlrarr{\chn_{i+2}}
\cdots
\rseqlrarr{\chn_{d-1}}\cyc_{d}\dashrightarrow\chn_{d}\]
be the representative for $[b,d]\in\Pers_*(\Fcal)$.
Then, the representative for $[b,d]\in\Pers_*(\Fcal')$
is updated to the following:
\[
\cdots
\rseqlrarr{\chn_{i-3}}\cyc_{i-2}
\rseqlrarr{\chn_{i-2}}\cyc'_{i-1}
\rseqlarr{0}\cyc_{i}
\rseqrarr{0}\cyc'_{i+1}
\rseqlrarr{\chn'_{i+1}}\cyc_{i+2}
\rseqlrarr{\chn_{i+2}}
\cdots,
\]
where $\cyc'_{i-1}=\cyc'_{i+1}:=\cyc_{i}$,
$\chn'_{i+1}:=\chn_{i}$, and the remaining cycles/chains
are as in $\rseq$.

\subsubsection{$\Hm_*(K'_{i-1})\leftarrow \Hm_*(K'_i)$ is injective}

In this case, 
$\pinds(\Fcal')=\pinds(\Fcal)\union\Set{i+1}$ and
$\ninds(\Fcal')=\ninds(\Fcal)\union\Set{i-1}$.
In order to obtain $\Pers_*(\Fcal')$,
we need to find `pairings' for the death index $i-1$
and birth index $i+1$ in $\Fcal'$.
Let $\Set{I_j\given j\in\Bcal}$ be the set of 
intervals in $\Pers_\Dim(\Fcal)$ containing $i$, where $\Bcal$
is an indexing set,
and let $\cyc^j_i$ be the representative $\Dim$-cycle at index $i$
for $I_j$.
Moreover, let $\LG:=\Set{j\in\Bcal\given\sG\in\cyc^j_i}$,
and for each $j\in\LG$, let $\tilde{\rseq}_j$ be the $\Dim$-th representative sequence
for $I_j$.
We do the following: 
\begin{itemize}
    \item Whenever there exist $j,k\in\LG$
s.t.\ $I_j\iles I_k$, 
update the representative for $I_k$ 
as $\tilde{\rseq}_j\repsum{\tilde{\rseq}}_k$,
and delete $k$ from $\LG$.
Note that the $\Dim$-cycle at index $i$ in $\tilde{\rseq}_j\repsum{\tilde{\rseq}}_k$
does not contain $\sG$.
\end{itemize}

After the above operations, we have that no two intervals in 
$\Set{I_j\given j\in\LG}$ are comparable. 
We then rewrite the intervals in $\Set{I_j\given j\in\LG}$ as 
\[[b_1,d_1],[b_2,d_2],\ldots,[b_\ell,d_\ell]\text{ s.t.\ }b_1\bles b_2\bles\cdots\bles b_\ell.\]

Also, for each $j$, let $\rseq_j$ be the $\Dim$-th representative sequence
for $[b_j,d_j]\in\Pers_*(\Fcal)$.

For $j\leftarrow 1,\ldots,\ell-1$, we do the following:
\begin{itemize}
    \item Note that $d_{j+1}\dles d_j$
because otherwise $[b_j,d_j]$ and $[b_{j+1},d_{j+1}]$ would be comparable.
Then, let $[b_{j+1},d_j]$ form an interval in $\Pers_*(\Fcal')$.
The representative is set as follows:
since $\rseq_j\repsum\rseq_{j+1}$ is a representative for $[b_{j+1},d_j]$
in $\Fcal$,
in which the $\Dim$-cycle at index $i$ does not contain $\sG$,
$\rseq_j\repsum\rseq_{j+1}$ can be `expanded' to become 
a representative for $[b_{j+1},d_j]\in\Pers_*(\Fcal')$
as done in Section~\ref{sec:out-expan-surj}.
\end{itemize}

After this, let $[b_1,i-1]$ and $[i+1,d_\ell]$ form two intervals in $\Pers_*(\Fcal')$
with representatives $\rseq_1\rsprefix{i}$ and $\rseq_\ell\rssuffix{i}$
respectively.

Finally,
all the intervals in $\Pers_*(\Fcal)$
that are not `touched' in the previous steps
are carried into $\Pers_*(\Fcal')$.
The updates of representatives for these intervals
remain the same as described in Section~\ref{sec:out-expan-surj}.

\subsubsection{Time complexity}
\label{sec:out-expan-complexity}
Determining injectivity/surjectivity at the beginning takes $O(\filtcnt+\simpcnt\log \simpcnt)$
time.
Representative update for each interval containing $i$
in Section~\ref{sec:out-expan-surj}
takes $O(\filtcnt)$ time,
and there are no more than $\simpcnt$ intervals containing $i$,
so the total time spent on the surjective case is $O(\filtcnt\simpcnt)$.
The bottleneck of the injective case is the two loops,
both of which take $O(\filtcnt\simpcnt^2)$ time.
Hence, the outward expansion takes $O(\filtcnt\simpcnt^2)$ time.

\subsection{Outward contraction}
\label{sec:out-contra}

Recall that an outward contraction is the following operation:

\vspace{-1em}
\begin{center}
\begin{tikzpicture}
\node (a) at (0,0) {$\Fcal: K_0 \leftrightarrow
\cdots\leftrightarrow K_{i-2}
\leftrightarrow 
K_{i-1}\bakinctosp{\sG} 
K_i 
\inctosp{\sG} K_{i+1}
\leftrightarrow K_{i+2}\leftrightarrow
\cdots \leftrightarrow K_\filtcnt$}; 
\node (b) at (0,-0.6){$\Fcal': K_0 \leftrightarrow
\cdots\leftrightarrow K_{i-2}
\leftrightarrow 
K'_{i}
\leftrightarrow K_{i+2}\leftrightarrow
\cdots \leftrightarrow K_\filtcnt$};
\draw[->,arrows={-latex},dashed] (a.0) .. controls (+7,-0.2) and (+6.5,-0.6) .. (b.0);
\end{tikzpicture}
\end{center}

\vspace{-1em}\noindent
where $K'_{i}=K_{i-1}=K_{i+1}$.
We also assume that  $\sG$ is a $\Dim$-simplex.
Notice that the indices for $\Fcal'$ are not consecutive,
i.e., $i-1$ and $i+1$ are skipped.

For the update,
we first determine whether the induced map 
$\Hm_*(K_{i-1})\leftarrow \Hm_*(K_i)$
is injective
or 
surjective
by checking whether 
$i-1$ is a death index in $\Fcal$
(injective)
or $i$ is a birth index in $\Fcal$
(surjective).

\subsubsection{$\Hm_*(K_{i-1})\leftarrow \Hm_*(K_i)$ is surjective}
\label{sec:out-contra-surj}
Since outward contractions are
inverses of outward expansions (see Section~\ref{sec:out-expan}),
the only difference of $\Pers_*(\Fcal)$ and $\Pers_*(\Fcal')$
in this case 
is that 
$[i,i]\in\Pers_*(\Fcal)$ is deleted in $\Pers_*(\Fcal')$.
Let $[b,d]\neq [i,i]$ be an interval in $\Pers_*(\Fcal)$.
If $i\not\in[b,d]$, i.e., $b>i$ or $d<i$, then since $b\neq i+1$
and $d\neq i-1$, we have that $b\geq i+2$ or $d\leq i-2$.
So the representative for $[b,d]\in\Pers_*(\Fcal)$
can be directly used as a representative for $[b,d]\in\Pers_*(\Fcal')$.
If $i\in[b,d]$, then 
suppose that 
\[\chn_{b-1}\dashleftarrow\cyc_{b}\rseqlrarr{\chn_{b}}\cdots
\rseqlrarr{\chn_{d-1}}\cyc_{d}\dashrightarrow\chn_{d}\]
is the representative for $[b,d]\in\Pers_*(\Fcal)$, 
which needs to be updated to the following for $[b,d]\in\Pers_*(\Fcal')$:
\[\chn_{b-1}\dashleftarrow\cyc_{b}\rseqlrarr{\chn_{b}}\cdots
\rseqlrarr{\chn_{i-3}}\cyc_{i-2}
\rseqlrarr{\chn_{i-2}+\chn_{i-1}}\cyc_{i}
\rseqlrarr{\chn_{i}+\chn_{i+1}}\cyc_{i+2}
\rseqlrarr{\chn_{i+2}}
\cdots
\rseqlrarr{\chn_{d-1}}\cyc_{d}\dashrightarrow\chn_{d}.\]

\subsubsection{$\Hm_*(K_{i-1})\leftarrow \Hm_*(K_i)$ is injective}

\paragraph{Step I.}
In this case, $i-1\in\ninds(\Fcal)$, $i+1\in\pinds(\Fcal)$,
$\ninds(\Fcal')=\ninds(\Fcal)\setminus\Set{i-1}$,
and $\pinds(\Fcal')=\pinds(\Fcal)\setminus\Set{i+1}$.
Let $[b_*,i-1]$ and $[i+1,d_\circ]$ be the $\Dim$-th intervals in $\Pers_*(\Fcal)$
ending/starting with $i-1$, $i+1$ respectively,
which have the following representatives:
\begin{align*}
\begin{split}
\rseq_*:\chn^*_{b_*-1}\dashleftarrow\cyc^*_{b_*}\rseqlrarr{\chn^*_{b_*}}\cdots
\rseqlrarr{\chn^*_{i-2}}\cyc^*_{i-1},
\\
\rseq_\circ:
\cyc^\circ_{i+1}\rseqlrarr{\chn^\circ_{i+1}}\cdots
\rseqlrarr{\chn^\circ_{d_\circ-1}}\cyc^\circ_{d_\circ}\dashrightarrow\chn^\circ_{d_\circ}.
\end{split}
\end{align*}
Then, let $\Set{[\bG_j,\dG_j]\given j\in\Bcal}$ be the set of intervals in $\Pers_\Dim(\Fcal)$
containing $i+1$, where $\Bcal$ is an indexing set.
Notice that $[i+1,d_\circ]\in\Set{[\bG_j,\dG_j]\given j\in\Bcal}$.
Moreover, for each $j\in\Bcal$, 
denote the $\Dim$-th representative for $[\bG_j,\dG_j]$ as:
\[\tilde{\rseq}_j:\tilde{\chn}^j_{\bG_j-1}\dashleftarrow\tilde{\cyc}^j_{\bG_j}
\rseqlrarr{\tilde{\chn}^j_{\bG_j}}\cdots
\rseqlrarr{\tilde{\chn}^j_{\dG_j-1}}\tilde{\cyc}^j_{\dG_j}\dashrightarrow\tilde{\chn}^j_{\dG_j}.\]
Then, the set of homology classes 
$\Set{[\tilde{\cyc}^j_{i+1}]\given {j\in\Bcal}}$,
which contains $[\cyc^\circ_{i+1}]$,
is a basis for $\Hm_\Dim(K_{i+1})$.
Since $\cyc^*_{i-1}\subseteq K_{i-1}=K_{i+1}$,
we can write $\cyc^*_{i-1}$ as the following sum:
\begin{equation}\label{eqn:z*i-1-sum}
\cyc^*_{i-1}=\cyc^\circ_{i+1}+\sum_{j\in\LG}\tilde{\cyc}^j_{i+1}+C,
\end{equation}
where $\LG\subseteq\Bcal$ and $C$ is the boundary
of a $(p+1)$-chain $A$ in $K_{i+1}$.
The sum in Equation~(\ref{eqn:z*i-1-sum}) must contain
$\cyc^\circ_{i+1}$ because: (i) $\sG\in\cyc^*_{i-1}$ and $\sG\in\cyc^\circ_{i+1}$
by Definition~\ref{dfn:rep-seq};
(ii) no cycle in $\Set{\tilde{\cyc}^j_{i+1}\given j\in\Bcal}$ other than $\cyc^\circ_{i+1}$ contains $\sG$
(Proposition~\ref{prop:homolog-cyc-no-new-simp});
(iii) no boundary in $K_{i+1}$ contains $\sG$ since $\sG$ has no cofaces in $K_{i+1}$.
Equation~(\ref{eqn:z*i-1-sum}) can be executed by first 
computing a boundary basis for $K_{i+1}$, 
which forms a cycle basis for $K_{i+1}$ along with $\Set{\tilde{\cyc}^j_{i+1}\given j\in\Bcal}$, and then performing a Gaussian elimination on the cycle basis.

\paragraph{Step II.}
Do the following: 
\begin{itemize}
    \item Whenever there is a $j\in\LG$
s.t.\ $\bG_j\bles b_*$, 
update the representative for $[b_*,i-1]$ 
as $\rseq_*:=\rseq_*\repsum\tilde{\rseq}_j$.
The update of $\rseq_*$ is valid because
$\dG_j\dles i-1$ ($\dG_j>i-1$
and $K_{i-1}\bakincto K_i$ is backward),
which means that $[\bG_j,\dG_j]\iles[b_*,i-1]$.
Then, delete $j$ from $\LG$.
    \item Similarly,
whenever there is a $j\in\LG$
s.t.\ $\dG_j\dles d_\circ$,
update the representative for $[i+1,d_\circ]$ 
as $\rseq_\circ:=\rseq_\circ\repsum\tilde{\rseq}_j$
because $[\bG_j,\dG_j]\iles[i+1,d_\circ]$.
Then, delete $j$ from $\LG$.
\end{itemize}

Note that Equation~(\ref{eqn:z*i-1-sum}) still holds 
after the above operations.
To see this, suppose that, e.g., there is an $\ell\in\LG$
s.t.\ $\bG_\ell\bles b_*$. We can rewrite
Equation~(\ref{eqn:z*i-1-sum}) as:
\begin{equation*}
\cyc^*_{i-1}+\tilde{\cyc}^\ell_{i-1}=\cyc^\circ_{i+1}+\sum_{j\in\LG\setminus\Set{\ell}}\tilde{\cyc}^j_{i+1}+C+\partial\big(\tilde{\chn}^\ell_{i-1}+\tilde{\chn}^\ell_{i}\big),
\end{equation*}
in which 
$\tilde{\cyc}^\ell_{i-1}=\tilde{\cyc}^\ell_{i+1}+\partial(\tilde{\chn}^\ell_{i-1}+\tilde{\chn}^\ell_{i})$.
Since $\cyc^*_{i-1}+\tilde{\cyc}^\ell_{i-1}$
is the cycle at index $i-1$ for the updated $\rseq_*$ in the iteration,
Equation~(\ref{eqn:z*i-1-sum}) still holds;
but we also need to update $C$ and $A$ in this case.

After the operations in this step, we have that $b_*\bles\bG_j$
and $d_\circ\dles\dG_j$ for any $j\in\LG$.

\paragraph{Step III.}
Rewrite the intervals in $\Set{[\bG_j,\dG_j]\given j\in\LG}$
as 
\[[b_1,d_1],[b_2,d_2],\ldots,[b_\ell,d_\ell]\text{ s.t.\ }b_1\bles b_2\bles\cdots\bles b_\ell.\]

Also, for each $j$ s.t.\ $1\leq j\leq\ell$, 
denote the $\Dim$-th representative for $[b_j,d_j]$ as
\[{\rseq}_j:{\chn}^j_{b_j-1}\dashleftarrow{\cyc}^j_{b_j}
\rseqlrarr{{\chn}^j_{b_j}}\cdots
\rseqlrarr{{\chn}^j_{d_j-1}}{\cyc}^j_{d_j}\dashrightarrow{\chn}^j_{d_j}.\]

Then, Equation~(\ref{eqn:z*i-1-sum}) can be rewritten as
\begin{equation}
\label{eqn:z*i-1-sum-rw}
\cyc^*_{i-1}=\cyc^\circ_{i+1}+\sum_{j=1}^\ell{\cyc}^j_{i+1}+C.
\end{equation}

Next, we pair the birth indices ${b_*,b_1,\ldots,b_\ell}$
with the death indices ${d_\circ,d_1,\ldots,d_\ell}$
to form intervals for $\Pers_*(\Fcal')$.
Initially, all these indices are `unpaired'.
We first pair $b_*$ with ${d}^\star=\Max_{\dles}\Set{d_1,\ldots,d_\ell}$
(and hence $b_*$, ${d^\star}$ become `paired')
to form an interval $[b_*,{d}^\star]\in\Pers_*(\Fcal')$,
with the following representative:
\begin{equation}
\label{eqn:b*-d-star-rep}
\rseq_*\rconcat\big(\rseq_\circ\drepsum\rseq_1\rssuffix{i+1}\drepsum
\cdots\drepsum\rseq_\ell\rssuffix{i+1}\big).
\end{equation}
We treat $\rseq_*$ as a $\Dim$-th post-birth representative
for $[b_*,i]$ in $\Fcal'$
and treat $\rseq_\circ\drepsum\rseq_1\rssuffix{i+1}\drepsum
\cdots\drepsum\rseq_\ell\rssuffix{i+1}$
as a $\Dim$-th pre-death representative
for $[i,d^\star]$ in $\Fcal'$
(because $d_\circ\dles d^\star$ and ${d}^\star=\Max_{\dles}\Set{d_1,\ldots,d_\ell}$).
The concatenation in Equation~(\ref{eqn:b*-d-star-rep})
is well-defined because (i) $\cyc^*_{i-1}$ is the $\Dim$-cycle at index $i-1$ in $\rseq_*$;
(ii) $\cyc^\circ_{i+1}+\sum_{j=1}^\ell{\cyc}^j_{i+1}$ is
the $\Dim$-cycle at index $i+1$ in $\rseq_\circ\drepsum\rseq_1\rssuffix{i+1}\drepsum
\cdots\drepsum\rseq_\ell\rssuffix{i+1}$;
(iii) the two $\Dim$-cycles are homologous in $K'_i=K_{i-1}=K_{i+1}$
due to Equation~(\ref{eqn:z*i-1-sum-rw}).

Similarly, we pair $b_\ell=\Max_{\bles}\Set{b_1,\ldots,b_\ell}$ with $d_\circ$
to form an interval $[b_\ell,d_\circ]\in\Pers_*(\Fcal')$,
with the following representative:
\begin{equation}
\label{eqn:b_l-d^circ-rep}
\big(\rseq_*\brepsum\rseq_1\rsprefix{i-1}\brepsum
\cdots\brepsum\rseq_\ell\rsprefix{i-1}\big)\rconcat\rseq_\circ.
\end{equation}

Then,
we pair the 
remaining indices
$\Set{b_1,\ldots,b_{\ell-1}}$
with
$\Set{d_1,\ldots,d_{\ell}}\setminus\Set{d^\star}$.
Specifically,
for $r:=1,\ldots,\ell-1$,
pair $b_r$ with a death index as follows.
\begin{itemize}
    \item If $d_r$ is unpaired, then pair $b_r$ with $d_r$.
    The representative for $[b_r,d_r]\in\Pers_*(\Fcal')$
    can be updated from the representative for $[b_r,d_r]\in\Pers_*(\Fcal)$
    as described in Section~\ref{sec:out-contra-surj}.
    
    \item If $d_r$ is paired, then 
    $d_\circ,d_1,\ldots,d_r$
    must be all the paired death indices so far
    because
    (i) $d_1,\ldots,d_{r-1}$ must be paired in previous iterations;
    (ii) the paired birth indices so far are $b_*,b_1,\ldots,b_{r-1},b_\ell$,
    which match the cardinality of 
    $d_\circ,d_1,\ldots,d_r$,
    and so there can be no more paired death indices.    
    Since $d_{r+1},\ldots,d_\ell$ are all unpaired, 
    we pair $b_r$ with $\dG=\Max_{\dles}\Set{d_{r+1},\ldots,d_\ell}$.
    The representative for $[b_r,\dG]\in\Pers_*(\Fcal)'$ is set as
    \begin{equation}
    \label{eqn:b_r-dG-rep}
    \big(\rseq_*\brepsum\rseq_1\rsprefix{i-1}\brepsum
    \cdots\brepsum\rseq_r\rsprefix{i-1}\big)\rconcat
    \big(\rseq_\circ\drepsum\rseq_{r+1}\rssuffix{i+1}\drepsum
    \cdots\drepsum\rseq_\ell\rssuffix{i+1}\big).
    \end{equation}
    The validity of the above representative follows from:
    (i) $b_*\bles b_1\bles\cdots\bles b_r$;
    (ii) the concatenation is well-defined because by Equation~(\ref{eqn:z*i-1-sum-rw}), $\cyc^*_{i-1}+\sum_{j=1}^r{\cyc}^j_{i-1}$ is homologous to
    $\cyc^\circ_{i+1}+\sum_{j=r+1}^\ell{\cyc}^j_{i+1}$ in $K'_i$.
\end{itemize}

Note that in order to compute the representative in Equation~(\ref{eqn:b_r-dG-rep})
efficiently,
we maintain the sum $\rseq_*\brepsum\rseq_1\rsprefix{i-1}\brepsum
\cdots\brepsum\rseq_r\rsprefix{i-1}$ at each iteration,
by adding $\rseq_r\rsprefix{i-1}$ to the sum for the previous iteration.
Similarly,
we maintain the sum $\rseq_\circ\drepsum\rseq_{r+1}\rssuffix{i+1}\drepsum
\cdots\drepsum\rseq_\ell\rssuffix{i+1}$,
which is initially $\rseq_\circ\drepsum\rseq_{1}\rssuffix{i+1}\drepsum
\cdots\drepsum\rseq_\ell\rssuffix{i+1}$,
and add $\rseq_{r}\rssuffix{i+1}$ at each iteration.
Since each iteration only performs a constant number of 
sums and concatenations of representatives,
which take $O(\filtcnt\simpcnt)$ time,
the total time spent on computing
Equation~(\ref{eqn:b_r-dG-rep}) is $O(\filtcnt\simpcnt^2)$.

\paragraph{Step IV.}
Every interval in $\Pers_*(\Fcal)$
that is not `touched' in the previous steps
is carried into $\Pers_*(\Fcal')$.
The update of representatives for these intervals
are the same as in Section~\ref{sec:out-contra-surj}.

\subsubsection{Time complexity}
\label{sec:out-contra-complexity}
By a similar analysis as in Section~\ref{sec:out-expan-complexity},
the time spent on the surjective case is $O(\filtcnt\simpcnt)$.
For the injective case, the complexity of Step I
is dominated by the cost of boundary basis computation, 
which can be accomplished in $O(\simpcnt^3)$ time 
by invoking a persistence algorithm~\cite{cohen2006vines}.
In Step II,
there are at most $\simpcnt$ iterations and each iteration
takes $O(\filtcnt\simpcnt)$ time.
So Step II takes $O(\filtcnt\simpcnt^2)$ time.
Step III is dominated by the computation of the representatives in
Equation~(\ref{eqn:b*-d-star-rep})$-$(\ref{eqn:b_r-dG-rep}),
which takes $O(\filtcnt\simpcnt^2)$ time.
The time taken in Step IV is the same as in the surjective case.
Hence, the outward contraction takes $O(\filtcnt\simpcnt^2)$ time.

\subsection{Forward switch}
\label{sec:fwd-switch}

Recall that a forward switch is the following operation:

\vspace{-1em}
\begin{center}
\begin{tikzpicture}
\tikzstyle{every node}=[minimum width=24em]
\node (a) at (0,0) {$\Fcal:K_0 \leftrightarrow\cdots\leftrightarrow K_{i-1}\inctosp{\sG}K_i\inctosp{\tG}K_{i+1}\leftrightarrow\cdots\leftrightarrow K_\filtcnt$}; 
\node (b) at (0,-0.6){$\Fcal':K_0\leftrightarrow\cdots\leftrightarrow K_{i-1}\inctosp{\tG} K'_i\inctosp{\sG} K_{i+1}\leftrightarrow\cdots\leftrightarrow K_\filtcnt$};
\path[->] (a.0) edge [bend left=90,looseness=1.5,arrows={-latex},dashed] (b.0);
\end{tikzpicture}
\end{center}

\vspace{-1em}\noindent
where $\sG\nsubseteq\tG$.
We have the following four cases, and 
the updating for each case is different:
\begin{enumerate}[label=\textbf{\Alph*}., ref=\Alph*]
    \item \label{itm:fwd-switch-bb}
    $K_{i-1}\inctosp{\sG}K_i$ provides a \emph{birth} index $i$ and 
    $K_{i}\inctosp{\tG}K_{i+1}$ provides a \emph{birth} index $i+1$ 
    in $\Fcal$.
    \item \label{itm:fwd-switch-dd}
    $K_{i-1}\inctosp{\sG}K_i$ provides a \emph{death} index $i-1$ and 
    $K_{i}\inctosp{\tG}K_{i+1}$ provides a \emph{death} index $i$
    in $\Fcal$.
    \item \label{itm:fwd-switch-bd}
    $K_{i-1}\inctosp{\sG}K_i$ provides a \emph{birth} index $i$ and 
    $K_{i}\inctosp{\tG}K_{i+1}$ provides a \emph{death} index $i$
    in $\Fcal$.
    \item \label{itm:fwd-switch-db}
    $K_{i-1}\inctosp{\sG}K_i$ provides a \emph{death} index $i-1$ and 
    $K_{i}\inctosp{\tG}K_{i+1}$ provides a \emph{birth} index $i+1$
    in $\Fcal$.
\end{enumerate}

\subsubsection{Case \ref{itm:fwd-switch-bb}}
\label{sec:fwd-switch-bb}

We have the following fact:

\begin{proposition}
By the assumptions of Case~\ref{itm:fwd-switch-bb}, one has that
$K_{i-1}\inctosp{\tG}K'_i$ provides a {birth} index $i$ and 
$K'_{i}\inctosp{\sG}K_{i+1}$ provides a {birth} index $i+1$ 
in $\Fcal'$.
\end{proposition}
\begin{proof}
Let $x\subseteq K_{i}$, $x'\subseteq K_{i+1}$ be cycles 
s.t.\ $\sG\in x$, $\tG\in x'$.
If $\sG\not\in x'$, then $\tG\in x'\subseteq K'_{i}$ and $\sG\in x\subseteq K_{i+1}$,
and hence the proposition is true.
If $\sG\in x'$, we can update $x'$ by summing it with $x$.
The new $x'$ satisfies that $x'\subseteq K_{i+1}$, $\tG\in x'$,
and $\sG\not\in x'$,
and hence the proposition is also true.
\end{proof}

\paragraph{Step I.}
An interval $[b,d]\in\Pers_*(\Fcal)$ s.t.\ $b\neq i,i+1$
is also an interval in $\Pers_*(\Fcal')$.
For updating its representative, we have the following situations:
\begin{description}
    \item[{$i\not\in[b,d]$}]:
    Since $b\neq i+1$, and $i-1$ is not a death index in $\Fcal$,
    we have that $b>i+1$ or $d<i-1$.
    So the representative for $[b,d]$ stays the same
    from $\Fcal$ to $\Fcal'$.
    \item[{$i\in[b,d]$}]: 
    Since $b\neq i$ and $i$ is not a death index in $\Fcal$,
    we have that $b\leq i-1$ and $d\geq i+1$.
    Let 
\[\tilde{\rseq}:\tilde{\chn}_{b-1}\dashleftarrow\tilde{\cyc}_{b}\rseqlrarr{\tilde{\chn}_{b}}\cdots
\rseqlrarr{\tilde{\chn}_{i-2}}\tilde{\cyc}_{i-1}
\rseqrarr{\tilde{\chn}_{i-1}}\tilde{\cyc}_{i}
\rseqrarr{\tilde{\chn}_{i}}\tilde{\cyc}_{i+1}
\rseqlrarr{\tilde{\chn}_{i+1}}\cdots
\rseqlrarr{\tilde{\chn}_{d-1}}\tilde{\cyc}_{d}\dashrightarrow\tilde{\chn}_{d}\]
be the representative for $[b,d]\in\Pers_*(\Fcal)$.
If $\sG\not\in\tilde{\chn}_{i-1}$
and $\sG\not\in\tilde{\cyc}_{i}$,
then $\tilde{\rseq}$ is still a representative for $[b,d]\in\Pers_*(\Fcal')$.
If $\sG\in\tilde{\chn}_{i-1}$
or $\sG\in\tilde{\cyc}_{i}$,
then let the following 
\[
\tilde{\chn}_{b-1}\dashleftarrow\tilde{\cyc}_{b}\rseqlrarr{\tilde{\chn}_{b}}\cdots
\rseqlrarr{\tilde{\chn}_{i-2}}\tilde{\cyc}_{i-1}
\rseqrarr{0}\tilde{\cyc}'_{i}
\rseqrarr{\tilde{\chn}_{i-1}+\tilde{\chn}_{i}}\tilde{\cyc}_{i+1}
\rseqlrarr{\tilde{\chn}_{i+1}}\cdots
\rseqlrarr{\tilde{\chn}_{d-1}}\tilde{\cyc}_{d}\dashrightarrow\tilde{\chn}_{d}\]
be the representative for $[b,d]\in\Pers_*(\Fcal')$,
where $\tilde{\cyc}'_{i}:=\tilde{\cyc}_{i-1}$
\end{description}

\paragraph{Step II.}
Let $[i,d_1],[i+1,d_2]$ be the intervals in $\Pers_*(\Fcal)$
starting with $i,i+1$ respectively, which have the following 
representatives:
\begin{alignat*}{2}
&\rseq_1:\cyc_{i}\rseqrarr{\chn_{i}}&\;&\cyc_{i+1}\rseqlrarr{\chn_{i+1}}\cdots
\rseqlrarr{\chn_{d_1-1}}\cyc_{d_1}\dashrightarrow\chn_{d_1},\\ 
&\rseq_2:&&\cyc'_{i+1}\rseqlrarr{\chn'_{i+1}}\cdots
\rseqlrarr{\chn'_{d_2-1}}\cyc'_{d_2}\dashrightarrow\chn'_{d_2}.
\end{alignat*}
Then, in order to obtain $\Pers_*(\Fcal')$,
we only need to pair the birth indices $i,i+1$ with the
death indices $d_1,d_2$ besides the intervals we inherit
directly in Step I.
By definition, we have that $\sG\in\cyc_{i}\subseteq K_i$
and $\tG\in\cyc'_{i+1}\subseteq K_{i+1}$.

If $\sG\not\in\cyc'_{i+1}$, 
then $[i+1,d_1]$ and $[i,d_2]$ form two intervals in $\Pers_*(\Fcal')$
with the following representatives:
\begin{alignat}{2}
&&\;&\cyc_{i+1}\rseqlrarr{\chn_{i+1}}\cdots
\rseqlrarr{\chn_{d_1-1}}\cyc_{d_1}\dashrightarrow\chn_{d_1},\\ 
\label{eqn:fwd-switch-case-A-prepend}
&\cyc'_{i}\rseqrarr{0}&&\cyc'_{i+1}\rseqlrarr{\chn'_{i+1}}\cdots
\rseqlrarr{\chn'_{d_2-1}}\cyc'_{d_2}\dashrightarrow\chn'_{d_2},
\end{alignat}
where $\cyc'_{i}:=\cyc'_{i+1}$.
It can be verified that the above representatives are valid.
For example, $\sG\in\cyc_{i+1}$ because $\cyc_{i+1}=\cyc_{i}+\partial(\chn_{i})$,
$\sG\in\cyc_{i}$, and $\sG\not\in\partial(\chn_{i})$ 
($\sG$ has no cofaces in $K_{i+1}$).

If $\sG\in\cyc'_{i+1}$, then we have the following situations
(note that $[i,d_1],[i+1,d_2]\in\Pers_*(\Fcal)$
are now of the same dimension):
\begin{description}
    \item[$d_1\dles d_2$]: Since $i\bles i+1$,
    we first update the representative
    for $[i+1,d_2]\in\Pers_*(\Fcal)$ as $\rseq_1\repsum\rseq_2$. 
    Note that $\sG\not\in\cyc_{i+1}+\cyc'_{i+1}$
    because $\sG\in\cyc_{i+1}$ as seen previously,
    where $\cyc_{i+1}+\cyc'_{i+1}$ is the cycle in $\rseq_1\repsum\rseq_2$
    at index $i+1$.
    With the updated representative
    for $[i+1,d_2]\in\Pers_*(\Fcal)$, 
    the rest of the operations are the same as done previously
    for $\sG\not\in\cyc'_{i+1}$.
    
    \item[$d_2\dles d_1$]:
    In this situation, the two intervals
    $[i,d_1],[i+1,d_2]\in\Pers_*(\Fcal)$
    are still intervals for $\Pers_*(\Fcal')$.
    The representative for $[i+1,d_2]\in\Pers_*(\Fcal')$
    is set to $\rseq_2$ because $\sG\in\cyc'_{i+1}$.
    The representative for $[i,d_1]\in\Pers_*(\Fcal')$
    is derived by prepending $\cyc_{i+1}+\cyc'_{i+1}\subseteq K'_i$ to 
    the beginning of
    $\rseq_1\repsum\rseq_2$ (which is defined over $[i+1,d_1]$),
    similarly to what is done to $\rseq_2$ in Equation~(\ref{eqn:fwd-switch-case-A-prepend});
    note that $\tG\in\cyc_{i+1}+\cyc'_{i+1}$ because
    $\tG\not\in\cyc_{i+1}$ (by Proposition~\ref{prop:homolog-cyc-no-new-simp})
    and $\tG\in\cyc'_{i+1}$.
\end{description}

\subsubsection{Case \ref{itm:fwd-switch-dd}}

We have the following fact:

\begin{proposition}
By the assumptions of Case~\ref{itm:fwd-switch-dd}, one has that
$K_{i-1}\inctosp{\tG}K'_i$ provides a {death} index $i-1$ and 
$K'_{i}\inctosp{\sG}K_{i+1}$ provides a {death} index $i$ 
in $\Fcal'$.
\end{proposition}
\begin{proof}
Since $\partial(\tG)$ is not a boundary in $K_i$,
$\partial(\tG)$ must not be a boundary in $K_{i-1}$,
and hence $K_{i-1}\inctosp{\tG}K'_i$ must provide a death index.
Now, for contradiction, suppose that $K'_{i}\inctosp{\sG}K_{i+1}$ provides a birth index,
i.e., there is a cycle $x\subseteq K_{i+1}$
s.t.\ $\sG\in x$.
Since $x\nsubseteq K_i$ (because otherwise 
$K_{i-1}\inctosp{\sG}K_i$ would have provided a birth index), 
we have that $\tG\in x$,
which contradicts the fact that $K_{i}\inctosp{\tG}K_{i+1}$ provides a {death} index.
\end{proof}

\paragraph{Step I.} 
An interval $[b,d]\in\Pers_*(\Fcal)$ s.t.\ $d\neq i-1,i$ 
is also an interval in $\Pers_*(\Fcal')$,
and the updating of representative for $[b,d]\in\Pers_*(\Fcal')$ is 
the same as in Step I for 
Case~\ref{itm:fwd-switch-bb} described in Section~\ref{sec:fwd-switch-bb}.

\paragraph{Step II.}
Let $[b_1,i-1],[b_2,i]$ be the intervals in $\Pers_*(\Fcal)$
ending with $i-1,i$ respectively, which have the following 
representatives:
\begin{alignat*}{1}
&\rseq_1:\chn_{b_1-1}\dashleftarrow\cyc_{b_1}\rseqlrarr{\chn_{b_1}}\cdots
\rseqlrarr{\chn_{i-2}}\cyc_{i-1}\dashrightarrow\chn_{i-1},\\
&\rseq_2:\chn'_{b_2-1}\dashleftarrow\cyc'_{b_2}\rseqlrarr{\chn'_{b_2}}\cdots
\rseqlrarr{\chn'_{i-2}}\cyc'_{i-1}
\rseqrarr{\chn'_{i-1}}\cyc'_{i}
\dashrightarrow\chn'_{i}.
\end{alignat*}
Then, in order to obtain $\Pers_*(\Fcal')$,
we only need to pair the birth indices $b_1,b_2$ with the
death indices $i-1,i$.
By definition, we have that $\sG\in\chn_{i-1}\subseteq K_i$
and $\tG\in\chn'_{i}\subseteq K_{i+1}$.

If $\sG\not\in\chn'_{i-1}+\chn'_{i}$,
then $[b_1,i],[b_2,i-1]$ form two intervals in $\Pers_*(\Fcal')$
with the following representatives:
\begin{alignat}{1}
&\chn_{b_1-1}\dashleftarrow\cyc_{b_1}\rseqlrarr{\chn_{b_1}}\cdots
\rseqlrarr{\chn_{i-2}}\cyc_{i-1}
\rseqrarr{0}\cyc_{i}\dashrightarrow\chn_{i},
\label{eqn:fwd-switch-dd-sg-notin-b1}\\
&\chn'_{b_2-1}\dashleftarrow\cyc'_{b_2}\rseqlrarr{\chn'_{b_2}}\cdots
\rseqlrarr{\chn'_{i-2}}\cyc'_{i-1}
\dashrightarrow\chn'_{i-1}+\chn'_{i},
\end{alignat}
where $\cyc_{i}:=\cyc_{i-1}$ and $\chn_{i}:=\chn_{i-1}$. 
It can be verified that the above representatives are valid. For example, 
we have that $\tG\in\chn'_{i-1}+\chn'_{i}\subseteq K'_i$ because 
$\tG\not\in\chn'_{i-1}\subseteq K_{i}$,
$\tG\in\chn'_{i}$,
and $\sG\not\in\chn'_{i-1}+\chn'_{i}\subseteq K_{i+1}$.

If $\sG\in\chn'_{i-1}+\chn'_{i}$,
then we have the following situations
(note that $[b_1,i-1],[b_2,i]\in\Pers_*(\Fcal)$
are now of the same dimension):
\begin{description}
    \item[$b_1\bles b_2$]:
    Now $[b_1,i],[b_2,i-1]$
    form two intervals for $\Pers_*(\Fcal')$.
    The representative for $[b_1,i]\in\Pers_*(\Fcal')$
    is the same as in
    Equation~(\ref{eqn:fwd-switch-dd-sg-notin-b1}).
    The representative for $[b_2,i-1]\in\Pers_*(\Fcal')$
    is derived from $\rseq_1\rsprefix{i-1}\brepsum\rseq_2\rsprefix{i-1}$
    by appending the chain $\chn_{i-1}+\chn'_{i-1}+\chn'_{i}$ to the end,
    where $\cyc_{i-1}+\cyc'_{i-1}=\partial(\chn_{i-1}+\chn'_{i-1}+\chn'_{i})$.
    Note that $\tG\in\chn_{i-1}+\chn'_{i-1}+\chn'_{i}\subseteq K'_{i}$ 
    because 
    $\tG\not\in\chn_{i-1}$,
    $\tG\not\in\chn'_{i-1}$,
    $\tG\in\chn'_{i}$,
    $\sG\in\chn_{i-1}$,
    $\sG\in\chn'_{i-1}+\chn'_{i}$,
    and $\chn_{i-1}+\chn'_{i-1}+\chn'_{i}\subseteq K_{i+1}$.
    
    \item[$b_2\bles b_1$]:
    In this situation, the two intervals
    $[b_1,i-1],[b_2,i]\in\Pers_*(\Fcal)$
    are still intervals for $\Pers_*(\Fcal')$.
    The representative for $[b_2,i]\in\Pers_*(\Fcal')$
    is:
    \[\chn'_{b_2-1}\dashleftarrow\cyc'_{b_2}\rseqlrarr{\chn'_{b_2}}\cdots
    \rseqlrarr{\chn'_{i-2}}\cyc'_{i-1}
    \rseqrarr{0}\cyc''_{i}
    \dashrightarrow\chn'_{i-1}+\chn'_{i},\]
    where $\cyc''_{i}:=\cyc'_{i-1}$ 
    and $\sG\in\chn'_{i-1}+\chn'_{i}$.
    The representative for $[b_1,i-1]\in\Pers_*(\Fcal')$
    is derived from $\rseq_1\rsprefix{i-1}\brepsum\rseq_2\rsprefix{i-1}$
    by appending the chain $\chn_{i-1}+\chn'_{i-1}+\chn'_{i}$ to the end.
\end{description}

\subsubsection{Case \ref{itm:fwd-switch-bd}}

We have the following fact:

\begin{proposition}
By the assumptions of Case~\ref{itm:fwd-switch-bd}, one has that
$K_{i-1}\inctosp{\tG}K'_i$ provides a {death} index $i-1$ and 
$K'_{i}\inctosp{\sG}K_{i+1}$ provides a {birth} index $i+1$ 
in $\Fcal'$.
\end{proposition}
\begin{proof}
Since $\partial(\tG)$ is not a boundary in $K_i$,
$\partial(\tG)$ must not be a boundary in $K_{i-1}$,
and hence $K_{i-1}\inctosp{\tG}K'_i$ must provide a death index.
Now, let $x\subseteq K_{i}$ be a cycle s.t.\ $\sG\in x$.
Then, $x$ is also in $K_{i+1}$,
and hence $K'_{i}\inctosp{\sG}K_{i+1}$ must provide a {birth} index.
\end{proof}

\paragraph{Step I.} 
An interval $[b,d]\in\Pers_*(\Fcal)$ s.t.\ $b\neq i$ and $d\neq i$
is also an interval in $\Pers_*(\Fcal')$,
and the updating of representative for $[b,d]\in\Pers_*(\Fcal')$ is 
the same as in Step I for 
Case~\ref{itm:fwd-switch-bb} described in Section~\ref{sec:fwd-switch-bb}.

\paragraph{Step II.}
Note that $[i,i]$ cannot form an interval in $\Pers_*(\Fcal)$.
To see this, suppose instead that $[i,i]\in\Pers_*(\Fcal)$. Then the fact 
that $\sG$ is in a boundary in $K_{i+1}$ (by Definition~\ref{dfn:rep-seq})
and $\sG$ has no cofaces in $K_i$ means that $\sG\subseteq\tG$,
which is a contradiction.
Let $[b,i]$ and $[i,d]$ be the intervals in $\Pers_*(\Fcal)$
ending and starting with $i$ respectively, which have the following 
representatives:
\begin{alignat*}{1}
&\rseq_1:\chn_{b-1}\dashleftarrow\cyc_{b}\rseqlrarr{\chn_{b}}\cdots
\rseqlrarr{\chn_{i-2}}\cyc_{i-1}
\rseqrarr{\chn_{i-1}}\cyc_{i}
\dashrightarrow\chn_{i},\\
&\rseq_2:\cyc'_{i}\rseqrarr{\chn'_{i}}
\cyc'_{i+1}\rseqlrarr{\chn'_{i+1}}\cdots
\rseqlrarr{\chn'_{d-1}}\cyc'_{d}\dashrightarrow\chn'_{d}.
\end{alignat*}
Then, $[b,i-1]$ and $[i+1,d]$ form intervals in $\Pers_*(\Fcal')$.
The representative for $[b,i-1]\in\Pers_*(\Fcal')$ is:
\[\chn_{b-1}\dashleftarrow\cyc_{b}\rseqlrarr{\chn_{b}}\cdots
\rseqlrarr{\chn_{i-2}}\cyc_{i-1}
\dashrightarrow\chn''_{i-1},\]
where $\chn''_{i-1}$ equals $\chn_{i-1}+\chn_{i}$ if $\sG\not\in\chn_{i-1}+\chn_{i}$
and equals $\chn_{i-1}+\chn_{i}+\cyc'_{i}$ otherwise.
The representative for $[i+1,d]\in\Pers_*(\Fcal')$ is:
\[\cyc'_{i+1}\rseqlrarr{\chn'_{i+1}}\cdots
\rseqlrarr{\chn'_{d-1}}\cyc'_{d}\dashrightarrow\chn'_{d},\]
where the proof for $\sG\in\cyc'_{i+1}$ is as done previously.

\subsubsection{Case \ref{itm:fwd-switch-db}}

We have the following fact:

\begin{proposition}\label{prop:fwd-switch-db-fact}
Given the assumptions of Case~\ref{itm:fwd-switch-db}, 
let $x\subseteq K_{i+1}$ be a cycle s.t.\ $\tG\in x$.
If $\sG\in x$, then 
$K_{i-1}\inctosp{\tG}K'_i$ provides a {death} index $i-1$ and 
$K'_{i}\inctosp{\sG}K_{i+1}$ provides a {birth} index $i+1$ 
in $\Fcal'$;
otherwise, 
$K_{i-1}\inctosp{\tG}K'_i$ provides a {birth} index $i$ and 
$K'_{i}\inctosp{\sG}K_{i+1}$ provides a {death} index $i$ 
in $\Fcal'$.
\end{proposition}
\begin{proof}
If $\sG\in x$, then $K'_{i}\inctosp{\sG}K_{i+1}$ must provide a {birth} index
because $x$ is a new cycle in $K_{i+1}$ created by the addition of $\sG$.
This implies that $\rank\Cyc_*(K_{i+1})=\rank\Cyc_*(K'_{i})+1$.
By the assumptions of Case~\ref{itm:fwd-switch-db},
we have that 
\begin{equation}\label{eqn:rank-rel}
\rank\Cyc_*(K_{i+1})=\rank\Cyc_*(K_{i-1})+1\text{ and }
\rank\Bnd_*(K_{i+1})=\rank\Bnd_*(K_{i-1})+1.
\end{equation}
So we must have that 
$\rank\Bnd_*(K'_{i})=\rank\Bnd_*(K_{i-1})+1$,
implying that 
$K_{i-1}\inctosp{\tG}K'_i$ provides a {death} index.
If $\sG\not\in x$, then $x\subseteq K'_i$ and is created by the
addition of $\tG$, implying that 
$K_{i-1}\inctosp{\tG}K'_i$ provides a {birth} index.
So for Equation~(\ref{eqn:rank-rel}) to hold,
$K'_{i}\inctosp{\sG}K_{i+1}$ must provide a {death} index.
\end{proof}

\paragraph{Step I.}
An interval $[b,d]\in\Pers_*(\Fcal)$ s.t.\ $b\neq i+1$ and $d\neq i-1$
is also an interval in $\Pers_*(\Fcal')$,
and the updating of representative for $[b,d]\in\Pers_*(\Fcal')$ is 
the same as in Step I for 
Case~\ref{itm:fwd-switch-bb} described in Section~\ref{sec:fwd-switch-bb}.

\paragraph{Step II.}
Let $[b,i-1]$ and $[i+1,d]$ be the intervals in $\Pers_*(\Fcal)$
ending with $i-1$ and starting with $i+1$ respectively, which have the following 
representatives:
\begin{alignat*}{1}
&\rseq_1:\chn_{b-1}\dashleftarrow\cyc_{b}\rseqlrarr{\chn_{b}}\cdots
\rseqlrarr{\chn_{i-2}}\cyc_{i-1}
\dashrightarrow\chn_{i-1},\\
&\rseq_2:
\cyc'_{i+1}\rseqlrarr{\chn'_{i+1}}\cdots
\rseqlrarr{\chn'_{d-1}}\cyc'_{d}\dashrightarrow\chn'_{d}.
\end{alignat*}
Note that $\tG\in\cyc'_{i+1}\subseteq K_{i+1}$.
By Proposition~\ref{prop:fwd-switch-db-fact},
the updating is different based on whether $\sG\in\cyc'_{i+1}$:

\begin{description}
\item[$\sG\in\cyc'_{i+1}$]:
We have that
$K_{i-1}\inctosp{\tG}K'_i$ provides a {death} index $i-1$ and 
$K'_{i}\inctosp{\sG}K_{i+1}$ provides a {birth} index $i+1$ 
in $\Fcal'$.
Since $\partial(\cyc'_{i+1})=0$ and $\sG,\tG\in\cyc'_{i+1}$, we have that 
$\partial(\sG)+\partial(\tG)=\partial(\cyc'_{i+1}\setminus\Set{\sG,\tG})$,
where $\cyc'_{i+1}\setminus\Set{\sG,\tG}\subseteq K_{i-1}$.
Hence, $\partial(\sG)$ is homologous to $\partial(\tG)$ in $K_{i-1}$,
and $[b,i-1]$ is still an interval in $\Pers_*(\Fcal')$
with the following representative:
\[\chn_{b-1}\dashleftarrow\cyc_{b}\rseqlrarr{\chn_{b}}\cdots
\rseqlrarr{\chn_{i-2}}\cyc_{i-1}
\dashrightarrow\chn_{i-1}+\cyc'_{i+1},\]
where $\tG\in\chn_{i-1}+\cyc'_{i+1}\subseteq K'_{i}$
and $\partial(\chn_{i-1}+\cyc'_{i+1})=\partial(\chn_{i-1})=\cyc_{i-1}$.
Also, since $\sG\in\cyc'_{i+1}$,
$[i+1,d]$ is still an interval in $\Pers_*(\Fcal')$
with $\rseq_2$ as a representative.

\item[$\sG\not\in\cyc'_{i+1}$]:
We have that
$K_{i-1}\inctosp{\tG}K'_i$ provides a {birth} index $i$ and 
$K'_{i}\inctosp{\sG}K_{i+1}$ provides a {death} index $i$ 
in $\Fcal'$.
Now, $[b,i],[i,d]$ form two intervals in $\Pers_*(\Fcal')$ 
with the following representatives:
\begin{alignat*}{1}
&\chn_{b-1}\dashleftarrow\cyc_{b}\rseqlrarr{\chn_{b}}\cdots
\rseqlrarr{\chn_{i-2}}\cyc_{i-1}
\rseqrarr{0}\cyc_{i}
\dashrightarrow\chn_{i},\\
&\cyc'_{i}\rseqrarr{0}
\cyc'_{i+1}\rseqlrarr{\chn'_{i+1}}\cdots
\rseqlrarr{\chn'_{d-1}}\cyc'_{d}\dashrightarrow\chn'_{d},
\end{alignat*}
where $\cyc_{i}:=\cyc_{i-1}$, $\chn_{i}:=\chn_{i-1}$,
and $\cyc'_{i}:=\cyc'_{i+1}$.
\end{description}

\subsubsection{Time complexity}

Step I of all the cases needs to go over $O(\filtcnt)$ intervals.
Since the representatives of only $O(\simpcnt)$ intervals need to be updated,
each of which takes $O(\simpcnt)$ time,
this step takes $O(\simpcnt^2+\filtcnt)$ time.
The bottleneck of Step II of all the cases is the addition of two
representative sequences, which takes $O(\filtcnt\simpcnt)$ time.
So the forward switch operation takes $O(\filtcnt\simpcnt)$ time.

\subsection{Outward switch}
\label{sec:owd-switch}

Recall that an outward switch is the following operation:

\vspace{-1em}
\begin{center}
\begin{tikzpicture}
\tikzstyle{every node}=[minimum width=24em]
\node (a) at (0,0) {$\Fcal: K_0 \leftrightarrow
\cdots
\leftrightarrow 
K_{i-1}\inctosp{\sG} 
K_i 
\bakinctosp{\tG} K_{i+1}
\leftrightarrow
\cdots \leftrightarrow K_\filtcnt$}; 
\node (b) at (0,-0.6){$\Fcal': K_0 \leftrightarrow
\cdots
\leftrightarrow 
K_{i-1}\bakinctosp{\tG} 
K'_i 
\inctosp{\sG} K_{i+1}
\leftrightarrow
\cdots \leftrightarrow K_\filtcnt$};
\path[->] (a.0) edge [bend left=90,looseness=1.5,arrows={-latex},dashed] (b.0);
\end{tikzpicture}
\end{center}

\vspace{-1em}\noindent
where $\sG\neq\tG$.
By the Mayer-Vietoris Diamond 
Principle~\cite{carlsson2010zigzag,carlsson2019parametrized,carlsson2009zigzag-realvalue},
there is an bijection between $\Pers_*(\Fcal)$ and $\Pers_*(\Fcal')$.
Let $[b,d]$ be an interval in $\Pers_*(\Fcal)$ with the following representatives:
\[\rseq:\chn_{b-1}\dashleftarrow\cyc_{b}\rseqlrarr{\chn_{b}}\cdots
\rseqlrarr{\chn_{d-1}}\cyc_{d}\dashrightarrow\chn_{d}.\]
We have seven different cases for $[b,d]\in\Pers_*(\Fcal)$ (see below).
In each case,
the form of the corresponding interval in $\Pers_*(\Fcal')$
and the updating of representative are different.
Note that the seven cases are disjoint and cover all the possibilities 
of $[b,d]$ because: (i)~Case~\ref{itm:owd-switch-bi-di}$-$\ref{itm:owd-switch-bi-dgeqi+1} correspond to $b=i$ or $d=i$ (which implies that $i\in[b,d]$);
(ii) Case~\ref{itm:owd-switch-bleqi-1-dgeqi+1} corresponds to $i\in[b,d]$ but $b\neq i$ and $d\neq i$;
(ii) the remaining cases correspond to $i\not\in[b,d]$.

\setcounter{desccounter}{0}
\begin{description}
\descitem{Case}{($b=i,d=i$)}:\label{itm:owd-switch-bi-di}
Suppose that $[b,d]\in\Pers_*(\Fcal)$ is 
in dimension $\Dim$.
The corresponding interval in $\Pers_{*}(\Fcal')$ is also
$[b,d]$ but in dimension $\Dim-1$.
The representative for $[b,d]\in\Pers_{\Dim-1}(\Fcal')$
is set to
$\chn'_{i-1}\dashleftarrow\cyc'_i\dashrightarrow\chn'_{i}$,
where $\cyc'_i=\partial(\tG)$, $\chn'_{i-1}=\tG$,
and $\chn'_{i}=\cyc_i\setminus\Set{\tG}$.

\descitem{Case}{($b<i,d=i$)}:\label{itm:owd-switch-bleqi-1-di}
The corresponding interval in $\Pers_{*}(\Fcal')$
is $[b,i-1]$ with the following representative:
\[\chn_{b-1}\dashleftarrow\cyc_{b}\rseqlrarr{\chn_{b}}\cdots
\rseqlrarr{\chn_{i-2}}\cyc_{i-1},\]
where $\tG\in\cyc_{i-1}$ because $\cyc_{i-1}=\cyc_{i}+\partial(\chn_{i-1})$,
$\tG\in\cyc_{i}$, and $\tG\not\in\partial(\chn_{i-1})$ 
($\tG$ has no cofaces in $K_i$).

\descitem{Case}{($b=i,d>i$)}:\label{itm:owd-switch-bi-dgeqi+1}
This case is symmetric to Case~\ref{itm:owd-switch-bleqi-1-di}
and the details are omitted.

\descitem{Case}{($b<i,d>i$)}:\label{itm:owd-switch-bleqi-1-dgeqi+1}
See Section~\ref{sec:owd-switch-bleqi-1-dgeqi+1}.

\descitem{Case}{($b=i+1$)}:\label{itm:owd-switch-bi+1}
The corresponding interval in $\Pers_{*}(\Fcal')$
is $[i,d]$. If $\sG\not\in\chn_{i}$, then $[i,d]\in\Pers_{*}(\Fcal')$
has the following representative:
\[\chn_{i-1}\dashleftarrow\cyc_{i}\rseqrarr{0}
\cyc_{i+1}\rseqlrarr{\chn_{i+1}}\cdots
\rseqlrarr{\chn_{d-1}}\cyc_{d}\dashrightarrow\chn_{d},\]
where $\cyc_{i}:=\cyc_{i+1}$
and $\chn_{i-1}:=\chn_{i}$.
Note that $\sG\not\in\partial(\chn_{i})=\cyc_{i+1}$ 
because $\sG$ has no cofaces in $K_i$, and hence
$\cyc_{i+1}\subseteq K'_i$.
If $\sG\in\chn_{i}$, then $[i,d]\in\Pers_{*}(\Fcal')$
has the following representative:
\[\chn_{i-1}\dashleftarrow\cyc_{i}\rseqrarr{\sG}
\cyc_{i+1}\rseqlrarr{\chn_{i+1}}\cdots
\rseqlrarr{\chn_{d-1}}\cyc_{d}\dashrightarrow\chn_{d},\]
where $\cyc_{i}:=\cyc_{i+1}+\partial(\sG)$
and $\chn_{i-1}:=\chn_{i}+{\sG}$.

\descitem{Case}{($d=i-1$)}:\label{itm:owd-switch-di-1}
This case is symmetric to Case~\ref{itm:owd-switch-bi+1}
and the details are omitted.

\descitem{Case}{($b>i+1$ or $d<i-1$)}:\label{itm:owd-switch-rest}
The corresponding interval in $\Pers_{*}(\Fcal')$
is $[b,d]$
and the representative
stays the same.

\end{description}

\subsubsection{Case \ref{itm:owd-switch-bleqi-1-dgeqi+1}}
\label{sec:owd-switch-bleqi-1-dgeqi+1}

In this case,
the corresponding interval in $\Pers_{*}(\Fcal')$
is still $[b,d]$.
If $\sG\not\in\chn_{i-1}$ and
$\tG\not\in\chn_{i}$, then the representative for 
$[b,d]\in\Pers_{*}(\Fcal')$ stays the same
besides the changes on the arrow directions.
For example, $\cyc_{i-1}\rseqrarr{\chn_{i-1}}\cyc_{i}$ in $\rseq$
now becomes $\cyc_{i-1}\rseqlarr{\chn_{i-1}}\cyc_{i}$
after the switch, 
where $\chn_{i-1}\subseteq K_{i-1}$.
Note that we always have $\cyc_{i}\subseteq K'_i$ because 
$\sG,\tG\not\in\cyc_{i}$ by Proposition~\ref{prop:homolog-cyc-no-new-simp}.

If $\sG\in\chn_{i-1}$ or
$\tG\in\chn_{i}$, then we have the following situations:
\begin{description}
\item[$\sG\not\in\chn_{i-1}+\chn_{i}$]: 
The representative for 
$[b,d]\in\Pers_{*}(\Fcal')$ is set to:
\[\chn_{b-1}\dashleftarrow\cyc_{b}\rseqlrarr{\chn_{b}}\cdots
\rseqlrarr{\chn_{i-2}}\cyc_{i-1}
\rseqlarr{\chn_{i-1}+\chn_{i}}\cyc'_{i}
\rseqrarr{0}\cyc_{i+1}
\rseqlrarr{\chn_{i+1}}\cdots
\rseqlrarr{\chn_{d-1}}\cyc_{d}\dashrightarrow\chn_{d},\]
where $\cyc'_{i}:=\cyc_{i+1}$,
$\cyc_{i-1}+\cyc_{i+1}=\partial(\chn_{i-1}+\chn_{i})$,
and $\chn_{i-1}+\chn_{i}\subseteq K_{i-1}$ because $\sG\not\in\chn_{i-1}+\chn_{i}$.
Note that $\cyc_{i+1}\subseteq K'_{i}$ because $\cyc_{i+1}$ as a cycle in $K_i$
does not contain $\sG$ by Proposition~\ref{prop:homolog-cyc-no-new-simp}.

\item[$\tG\not\in\chn_{i-1}+\chn_{i}$]: 
Symmetrically, the representative for 
$[b,d]\in\Pers_{*}(\Fcal')$ is set to:
\[\chn_{b-1}\dashleftarrow\cyc_{b}\rseqlrarr{\chn_{b}}\cdots
\rseqlrarr{\chn_{i-2}}\cyc_{i-1}
\rseqlarr{0}\cyc'_{i}
\rseqrarr{\chn_{i-1}+\chn_{i}}\cyc_{i+1}
\rseqlrarr{\chn_{i+1}}\cdots
\rseqlrarr{\chn_{d-1}}\cyc_{d}\dashrightarrow\chn_{d},\]
where $\cyc'_{i}:=\cyc_{i-1}$.

\item[$\tG,\sG\in\chn_{i-1}+\chn_{i}$]: 
The representative for 
$[b,d]\in\Pers_{*}(\Fcal')$ is set to:
\[\chn_{b-1}\dashleftarrow\cyc_{b}\rseqlrarr{\chn_{b}}\cdots
\rseqlrarr{\chn_{i-2}}\cyc_{i-1}
\rseqlarr{\chn_{i-1}+\chn_{i}+\sG}\cyc'_{i}
\rseqrarr{\sG}\cyc_{i+1}
\rseqlrarr{\chn_{i+1}}\cdots
\rseqlrarr{\chn_{d-1}}\cyc_{d}\dashrightarrow\chn_{d},\]
where $\cyc'_{i}:=\cyc_{i+1}+\partial(\sG)$.
\end{description}

\subsubsection{Time complexity}
Going over all the intervals in $\Pers_*(\Fcal)$
takes $O(\filtcnt)$ time,
and each case takes no more than $O(\simpcnt)$ time.
Since Case \ref{itm:owd-switch-bleqi-1-dgeqi+1} can be executed
for no more than $\simpcnt$ times, 
the time complexity of outward switch operation is $O(\simpcnt^2+\filtcnt)$.

\section{Conclusion} 
We have presented update algorithms 
for maintaining barcodes and/or representatives
over a changing zigzag filtration. Two
main questions ensue from this research: (i) Can we make the updates more efficient?
Six operations that can be implemented using transpositions
in converted non-zigzag filtrations cannot be improved unless their
non-zigzag analogues are improved. A big open question is whether the other two operations, namely 
outward expansion and contraction, can be done without maintaining representatives explicitly. Or, is it possible to maintain representatives explicitly with a better complexity?
(ii) Are there interesting applications of the update algorithms presented in this paper? We have mentioned their application to computing vineyards for dynamic point clouds and also to multiparameter persistence. We believe that there will be other dynamic settings where updating zigzag persistence plays a contributing role.

\bibliographystyle{plainurl}
\bibliography{refs}

\begin{thebibliography}{10}

\bibitem{agarwal2006extreme}
Pankaj~K. Agarwal, Herbert Edelsbrunner, John Harer, and Yusu Wang.
\newblock Extreme elevation on a 2-manifold.
\newblock {\em Discrete \& Computational Geometry}, 36(4):553--572, 2006.

\bibitem{carlsson2010zigzag}
Gunnar Carlsson and Vin de~Silva.
\newblock Zigzag persistence.
\newblock {\em Foundations of Computational Mathematics}, 10(4):367--405, 2010.

\bibitem{carlsson2019parametrized}
Gunnar Carlsson, Vin de~Silva, Sara Kali{\v{s}}nik, and Dmitriy Morozov.
\newblock Parametrized homology via zigzag persistence.
\newblock {\em Algebraic \& Geometric Topology}, 19(2):657--700, 2019.

\bibitem{carlsson2009zigzag-realvalue}
Gunnar Carlsson, Vin de~Silva, and Dmitriy Morozov.
\newblock Zigzag persistent homology and real-valued functions.
\newblock In {\em Proceedings of the Twenty-Fifth Annual Symposium on
  Computational Geometry}, pages 247--256, 2009.

\bibitem{cohen2009extending}
David Cohen-Steiner, Herbert Edelsbrunner, and John Harer.
\newblock Extending persistence using {P}oincar{\'e} and {L}efschetz duality.
\newblock {\em Foundations of Computational Mathematics}, 9(1):79--103, 2009.

\bibitem{cohen2006vines}
David Cohen-Steiner, Herbert Edelsbrunner, and Dmitriy Morozov.
\newblock Vines and vineyards by updating persistence in linear time.
\newblock In {\em Proceedings of the Twenty-Second Annual Symposium on
  Computational Geometry}, pages 119--126, 2006.

\bibitem{dey2021graph}
Tamal~K. Dey and Tao Hou.
\newblock Computing zigzag persistence on graphs in near-linear time.
\newblock In {\em 37th International Symposium on Computational Geometry, SoCG
  2021}, volume 189 of {\em LIPIcs}, pages 30:1--30:15. Schloss Dagstuhl -
  Leibniz-Zentrum f{\"{u}}r Informatik, 2021.

\bibitem{dey2022fast}
Tamal~K. Dey and Tao Hou.
\newblock Fast computation of zigzag persistence.
\newblock {\em 30th Annual European Symposium on Algorithms (ESA 2022), {\rm to
  appear}}, 2022.

\bibitem{DKM21}
Tamal~K. Dey, Woojin Kim, and Facundo M\'{e}moli.
\newblock Computing generalized rank invariant for 2-parameter persistence
  modules via zigzag persistence and its applications.
\newblock {\em https://arXiv.org/abs/2111.15058}, 2021.

\bibitem{edelsbrunner2012medusa}
Herbert Edelsbrunner, Carl-Philipp Heisenberg, Michael Kerber, and Gabriel
  Krens.
\newblock The medusa of spatial sorting: Topological construction.
\newblock {\em arXiv preprint arXiv:1207.6474}, 2012.

\bibitem{edelsbrunner2000topological}
Herbert Edelsbrunner, David Letscher, and Afra Zomorodian.
\newblock Topological persistence and simplification.
\newblock In {\em Proceedings 41st Annual Symposium on Foundations of Computer
  Science}, pages 454--463. IEEE, 2000.

\bibitem{Gabriel72}
Peter Gabriel.
\newblock {Unzerlegbare Darstellungen I}.
\newblock {\em Manuscripta Mathematica}, 6(1):71--103, 1972.

\bibitem{hatcher2002algebraic}
Allen Hatcher.
\newblock {\em Algebraic Topology}.
\newblock Cambridge University Press, 2002.

\bibitem{kim2020spatiotemporal}
Woojin Kim and Facundo M{\'e}moli.
\newblock Spatiotemporal persistent homology for dynamic metric spaces.
\newblock {\em Discrete \& Computational Geometry}, pages 1--45, 2020.

\bibitem{KM20}
Woojin Kim and Facundo M{\'e}moli.
\newblock Generalized persistence diagrams for persistence modules over posets.
\newblock {\em Journal of Applied and Computational Topology}, 5(4):533--581,
  2021.

\bibitem{maria2014zigzag}
Cl{\'e}ment Maria and Steve~Y. Oudot.
\newblock Zigzag persistence via reflections and transpositions.
\newblock In {\em Proceedings of the Twenty-Sixth Annual ACM-SIAM Symposium on
  Discrete Algorithms}, pages 181--199. SIAM, 2014.

\bibitem{maria2016computing}
Cl{\'e}ment Maria and Steve~Y. Oudot.
\newblock Computing zigzag persistent cohomology.
\newblock {\em arXiv preprint arXiv:1608.06039}, 2016.

\bibitem{maria2019discrete}
Cl{\'e}ment Maria and Hannah Schreiber.
\newblock Discrete morse theory for computing zigzag persistence.
\newblock In {\em Workshop on Algorithms and Data Structures}, pages 538--552.
  Springer, 2019.

\bibitem{milosavljevic2011zigzag}
Nikola Milosavljevi{\'c}, Dmitriy Morozov, and Primoz Skraba.
\newblock Zigzag persistent homology in matrix multiplication time.
\newblock In {\em Proceedings of the Twenty-Seventh Annual Symposium on
  Computational Geometry}, pages 216--225, 2011.

\bibitem{oudot2015zigzag}
Steve~Y. Oudot and Donald~R. Sheehy.
\newblock Zigzag zoology: Rips zigzags for homology inference.
\newblock {\em Foundations of Computational Mathematics}, 15(5):1151--1186,
  2015.

\bibitem{Patel}
Amit Patel.
\newblock Generalized persistence diagrams.
\newblock {\em Journal of Applied and Computational Topology}, 1(3):397--419,
  2018.

\bibitem{reynolds1987flocks}
Craig~W. Reynolds.
\newblock Flocks, herds and schools: A distributed behavioral model.
\newblock In {\em Proceedings of the 14th Annual Conference on Computer
  Graphics and Interactive Techniques}, pages 25--34, 1987.

\end{thebibliography}

\appendix

\section{Potential applications of zigzag update algorithms}
\label{sec:zzup-app}

\paragraph{Dynamic point cloud.}
Consider a set of points $P$ moving with respect to time~\cite{edelsbrunner2012medusa,kim2020spatiotemporal}.
For each point pair in $P$, we can draw its \emph{distance-time curve}
revealing the variation of distance between the points w.r.t.\ time.
For example, Figure~\ref{fig:grid_curve} draws the curves for a simple
$P$ with three points,
where $e_1,e_2$ and $e_3$ denote edges formed by the three point pairs.
Consider the Vietoris-Rips complex of $P$ with $\dG$ as the distance threshold.
Since distances of the point pairs may become greater or less than
$\dG$ at different time,
edges formed by these pairs
are added to or deleted from the Rips complex accordingly.
This forms a zigzag filtration of Rips complexes, which we denote as $\Rcal^\dG$.
Letting $\dG$ vary from $0$ to $\infty$, and taking the persistence diagram (PD)
of $\Rcal^\dG$, we obtain a vineyard~\cite{cohen2006vines} as a descriptor
for the dynamic point cloud.
We note that $\Rcal^\dG$ changes only
at the \emph{critical} points of the distance-time curves,
which are local minima/maxima and intersections
(as illustrated by the dots in Figure~\ref{fig:grid_curve}).
To compute the vineyard, one only needs to compute the PD
of each $\Rcal^\dG$ where $\dG$ is in between
distance values of two critical points.
For example, $\Set{\dG_i}_i$ are the distance values for the critical points
in Figure~\ref{fig:grid_curve},
and $\Set{d_i}_i$ are the values in between.
Figure~\ref{fig:grid_module} lists the zigzag filtration $\Rcal^{d_i}$
for each $d_i$, where each horizontal arrow is either an equality, addition of an edge,
or deletion of an edge.
Each transition from $\Rcal^{d_i}$ to $\Rcal^{d_{i+1}}$ can be realized by
a sequence of atomic operations described in this paper,
which provides natural associations for the PDs~\cite{cohen2006vines}.
For example, starting from the top and going down, 
one needs to perform forward/backward/outward switches,
inward contractions, and outward expansions
(defined in Section~\ref{sec:update-oper}).
One could also start from the bottom and go up,
which requires the reverse operations.
In Section~\ref{sec:dpc},
we provide details on how the zigzag filtrations are built
for a dynamic point cloud and how the atomic operations can 
be used to realize the transitions.

\begin{figure}[p]
  \centering
  \captionsetup{justification=centering}
  \captionsetup[subfigure]{justification=centering}

  \begin{subfigure}[t]{\textwidth}
  \centering
  \includegraphics[width=0.65\linewidth]{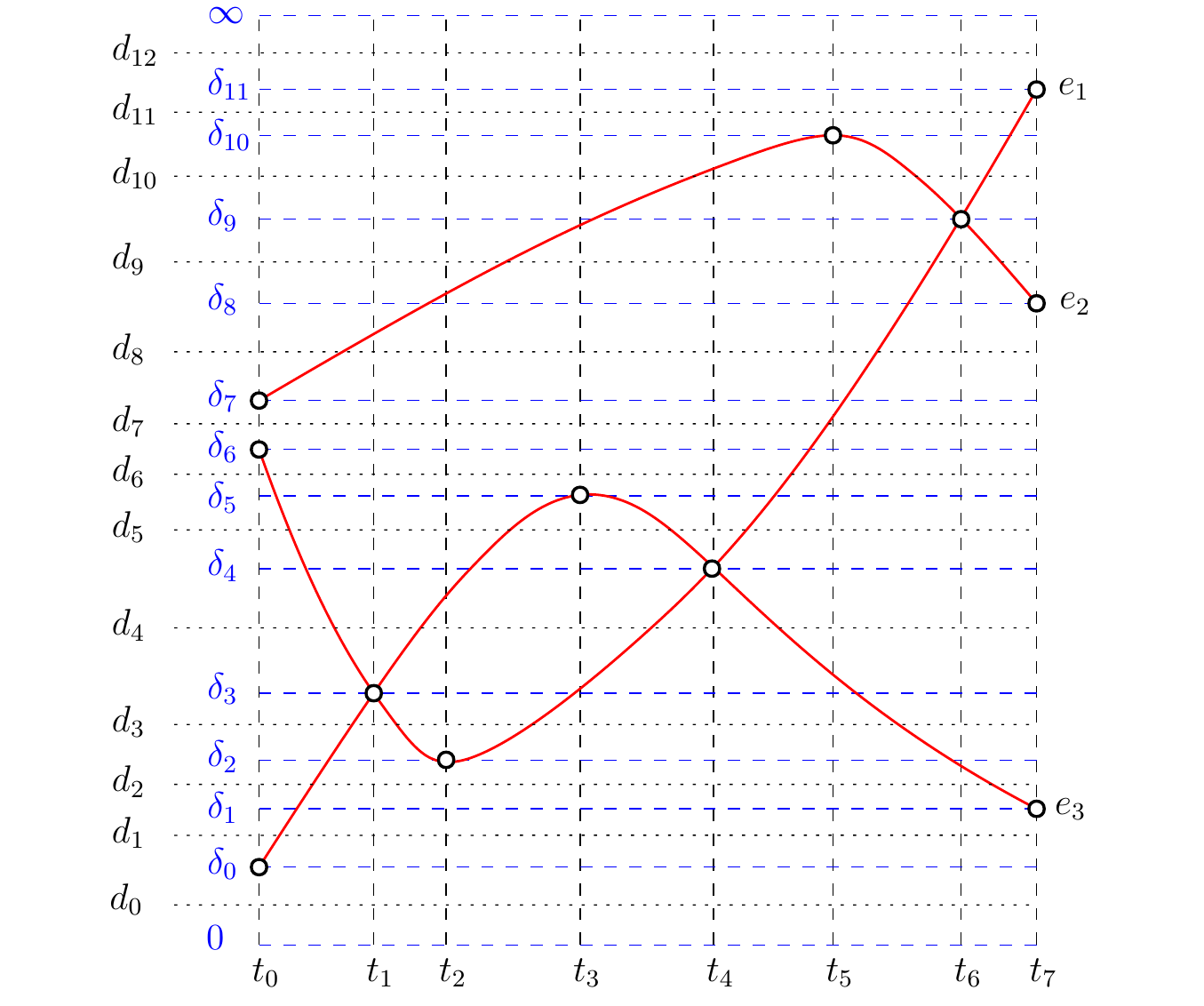}
  \caption{Distance-time curves of the three point pairs.}
   \label{fig:grid_curve}
  \end{subfigure}
  
  \vspace{2em}
   
  \captionsetup[subfigure]{justification=justified}

  \begin{subfigure}[t]{\textwidth}
  \centering
  \includegraphics[width=0.65\linewidth]{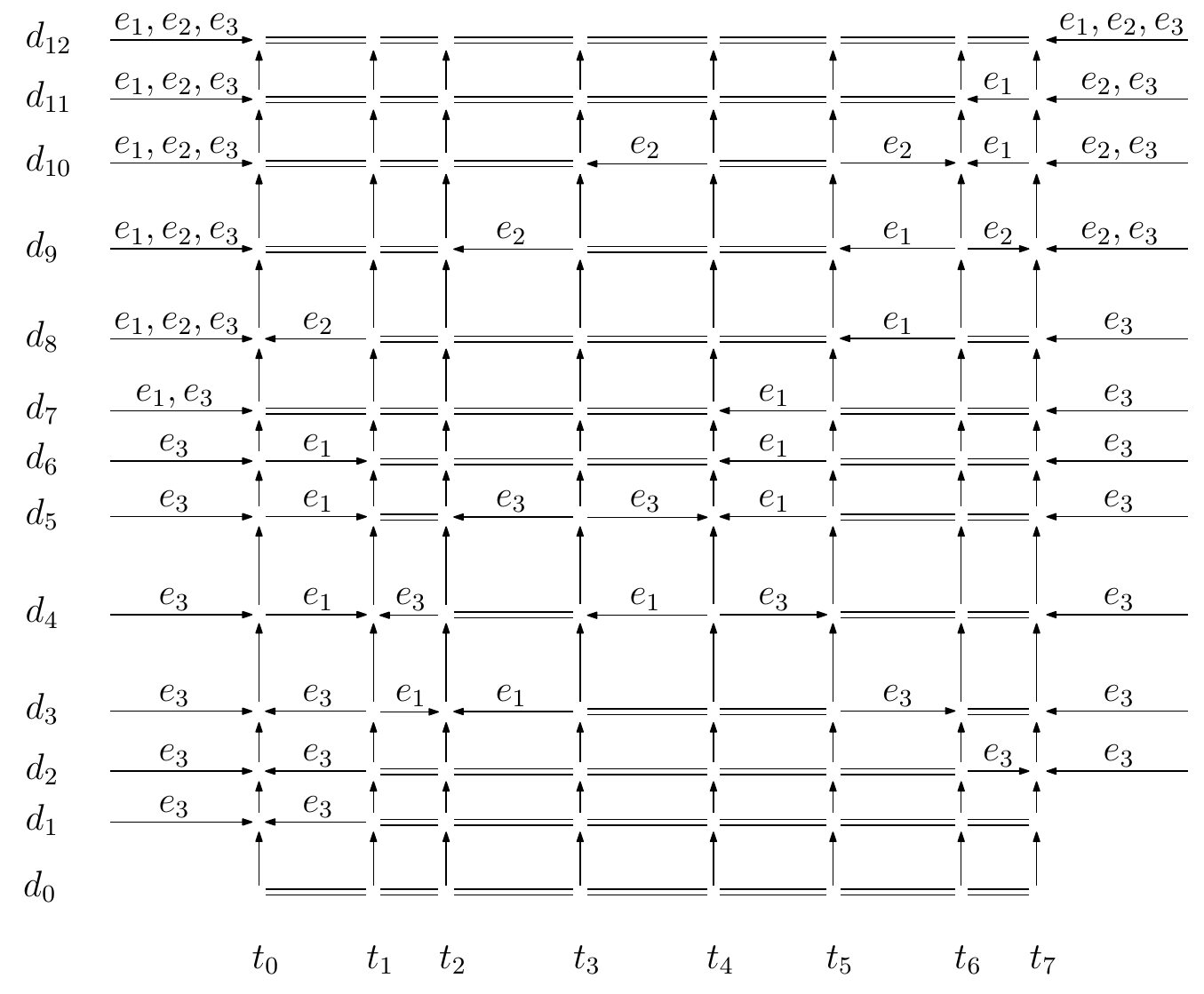}
  \caption{Zigzag filtration $\Rcal^{d_i}$ for each $d_i$ is listed horizontally, while vertically 
  each Rips complex is included into the one on the above.}
   \label{fig:grid_module}
  \end{subfigure}

  \caption{An example of a dynamic point cloud with three points.}
  \label{fig:dpc-grid}
\end{figure}

\paragraph{Levelset zigzag for time-varying function.}
It is known that the level sets
of a function give rise to a special type of zigzag filtrations
called levelset zigzag filtrations~\cite{carlsson2009zigzag-realvalue}, 
which are known to capture more information than the non-zigzag sublevel-set filtrations. Thus, even for a time-varying function, computing
the vineyard for a levelset zigzag filtration may capture more information than
the one by non-zigzag filtrations. 

\paragraph{Other potential applications.}
We also hope that our algorithms
for maintaining the representatives 
may be of independent interest.
For example, an efficient maintenance of these representatives provided
an efficient algorithm for computing zigzag persistence on graphs~\cite{dey2021graph} 
and also explained why a persistence algorithm proposed by Agarwal et al.~\cite{agarwal2006extreme} 
for elevation
functions works. Hilbert (dimension) function or rank 
function are among
some of the basic features for a
multiparameter persistence module.
One may use zigzag updates to compute these functions more efficiently
as Figure~\ref{fig:use2d}a suggests.
Thinking forward, we see a potential use of our algorithms
for maintaining representatives to compute generalized rank invariants~\cite{KM20,Patel}
for 2-parameter persistence modules. This may help
compute different homological structures as 
advocated recently~\cite{DKM21}; see Figure~\ref{fig:use2d}b.

\begin{figure}[htbp]
    \centering
    \begin{tabular}{cc}
    \includegraphics[height=4cm]{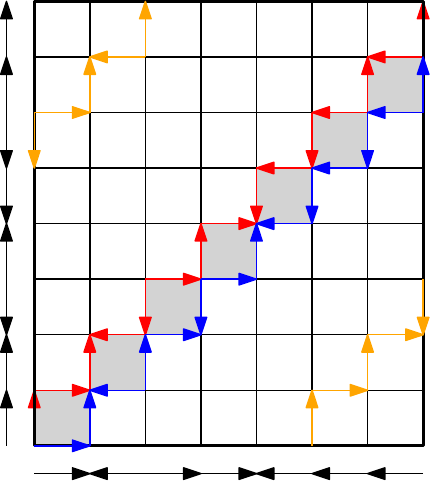}&
    \includegraphics[height=4cm]{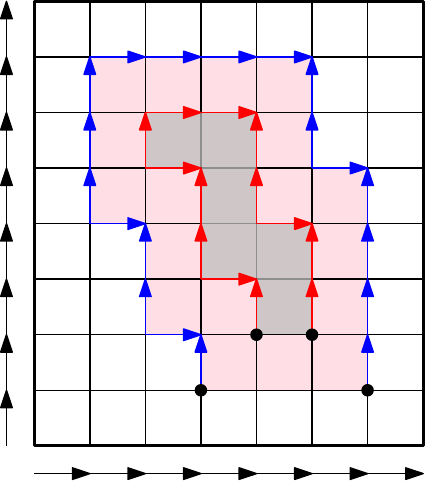}\\
    (a) & (b)
    \end{tabular}
        \caption{(a) Computing dimensions or rank function on a persistence module 
        with support over a 2D zigzag grid (poset) can be more efficiently computed by considering zigzag persistence on an initial zigzag filtration (indicated by red path) and then updating it with switches, which gives
        other zigzag paths (indicated by blue and golden paths). Assuming $t$ points
        in the grid, this will take $O(t^3)$ time with the updates instead of $O(t^{\omega+2})$
        with bruteforce zigzag persistence computation on every path. (b) Recently, it is
        shown that the generalized rank of an interval in a 2-parameter module can be derived
        from the zigzag persistence on the boundary as shown with red and blue paths for the
        grey and pink intervals respectively~\cite{DKM21}. We can leverage our update algorithms to compute
        the zigzag persistence over these two paths and multiple boundaries in general.}
    \label{fig:use2d}
\end{figure}

\subsection{Details on dynamic point clouds}
\label{sec:dpc}

We first define the following:

\begin{definition}
Throughout the section,
let $\dpc=(\pset,\posf_0,\posf_1,\ldots,\posf_\dpccnt)$
denote a dynamic point cloud in which: (i) $\pset$ is a set of points;
(ii) each map $\posf_i:\pset\to\Real^\Dim$ specifies the positions
of points in $\pset$ at \emph{time} $i$.
\end{definition}

Note that while 1-dimensional persistence~\cite{edelsbrunner2000topological}
with Rips filtration
serves as an effective descriptor for a fixed point cloud, 
it cannot naturally characterize a dynamic point cloud as defined above~\cite{kim2020spatiotemporal}.
In view of this,
we build vines and vineyards~\cite{cohen2006vines}
as descriptors for $\dpc$ using zigzag persistence.
We first let the time in $\dpc$ range {continuously}
in $[0,\dpccnt]$, i.e., the position of each point in $\pset$
during time $[0,\dpccnt]$ is linearly interpolated
based on 
the discrete samples given in $\dpc$.
For each $t\in[0,\dpccnt]$, let $\pset_t$ denote the point cloud
which is the point set $\pset$ with positions at time $t$.
Also, for a $\dG\geq 0$, let $R^\dG_t$ denote the Rips complex of $\pset_t$
with distance $\dG$.

Now fix a $\dG\geq 0$, and consider the continuous sequence 
$\Rcal^\dG:=\Set{R^\dG_t}_{t\in[0,\dpccnt]}$.
We claim that $\Rcal^\dG$ is encoded by a zigzag filtration,
and hence admits a barcode (persistence diagram)
as descriptor.
To see this, we note that each $R^\dG_t$ in $\Rcal^\dG$ is completely determined
by the vertex pairs in $\pset$ with distances no greater than $\dG$ at time $t$. 
Let $\pi$ be a vertex pair whose distance varies with time as illustrated
by the red curve in Figure~\ref{fig:dis_plot},
where the horizontal axis denotes time and the vertical axis denotes distance.
For the $\dG$ in Figure~\ref{fig:dis_plot},
the edge formed by $\pi$ is in $R^\dG_t$ when $t$ falls in the intervals
$[0,t_1]$, $[t_2,t_3]$, and $[t_4,t_5]$.
Also, in Figure~\ref{fig:pair_intervals}, for three vertex pairs $\pi_1,\pi_2,\pi_3$,
we illustrate respectively the time intervals in which 
their distances are no greater than $\dG$.
With the time varying, 
the edges formed by the vertex pairs are added to
or deleted from the Rips complex. 
As illustrated in Figure~\ref{fig:pair_intervals},
this naturally defines a zigzag filtration which we denote as $\Fcal^\dG$.
For example, $R_{t_2}^\dG$ in Figure~\ref{fig:pair_intervals}
is defined by edges formed by $\pi_2$ and $\pi_3$, and
$R_{t_5}^\dG$ 
is defined by edges formed by $\pi_1$ and $\pi_3$.

\begin{figure}[t]
  \centering
  \captionsetup[subfigure]{justification=centering}

  \begin{subfigure}[t]{0.4\textwidth}
  \centering
  \includegraphics[width=\linewidth]{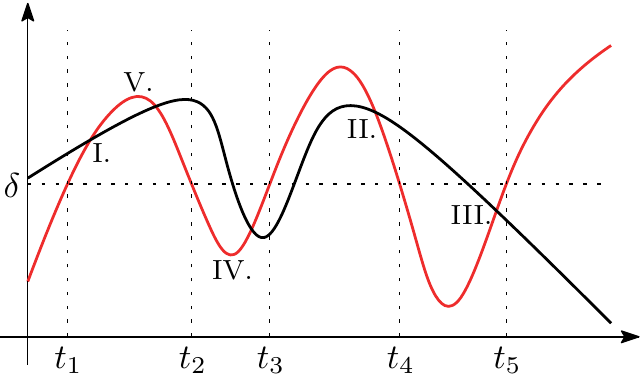}
  \caption{}
  \label{fig:dis_plot}
  \end{subfigure}
  \hspace{1.5em}
  \begin{subfigure}[t]{0.5\textwidth}
  \centering
  \includegraphics[width=\linewidth]{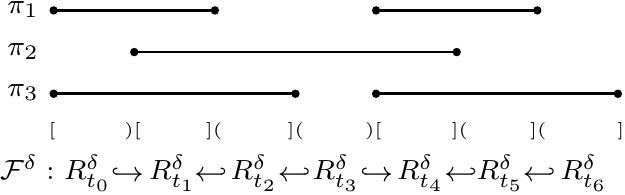}
  \caption{}
  \label{fig:pair_intervals}
  \end{subfigure}

  \caption{(a) Distance-time curves for two vertex pairs.
  (b) Time intervals for three vertex pairs $\pi_1,\pi_2,\pi_3$ in which distance is $\leq\dG$
  and the corresponding zigzag filtration $\Fcal^\dG$.}
\end{figure}

We then consider the one-parameter family of 
persistence diagrams
$\Set{\Bcal^\dG}_{\dG\in[0,\infty]}$,
with $\Bcal^\dG$ 
being the persistence diagram of 
$\Rcal^\dG$,
which forms a vineyard~\cite{cohen2006vines}. 
Treating each $\Bcal^\dG$ as a multi-set
of points in $\Real^2$,
the vineyard 
$\Set{\Bcal^\dG}_{\dG\in[0,\infty]}$
contains vines tracking the movement of points in persistence diagrams 
w.r.t.\ $\dG$.
For computing the vineyard $\Set{\Bcal^\dG}_{\dG\in[0,\infty]}$, 
we utilize
the update operations 
and algorithms presented in this paper.
As in~\cite{cohen2006vines},
our atomic update operations help associate points 
for persistence diagrams in $\Set{\Bcal^\dG}_{\dG\in[0,\infty]}$
without ambiguity, 
which is otherwise unavoidable if attempting to associate directly.
Let $\bar{\dG}$ be the maximum distance of vertex pairs at all time
in $\dpc$. We start with $\Rcal^{\bar{\dG}}$.
Since $R^{\bar{\dG}}_t$ equals a contractible (high-dimensional) simplex
at any $t$,
$\Bcal^{\bar{\dG}}$ contains only a 0-th interval $[0,\dpccnt]$
whose representative sequence
is straightforward\footnote{In practice,
  one may only consider simplices up to a dimension to save time;
  $\Bcal^{\bar{\dG}}$ and the representatives 
  in this case
  can then be computed from a homology basis for the complex at a time $t$.}.
Now consider the distance-time curves of \emph{all} vertex pairs of $\pset$
(e.g., Figure~\ref{fig:dis_plot} illustrates curves of two pairs),
which indeed defines a \emph{dynamic metric space}~\cite{kim2020spatiotemporal}.
When decreasing the distance $\dG$, 
$\Fcal^\dG$ changes only
at the following types of points
in the plot of all distance-time curves (see Figure~\ref{fig:dpc_events}):
\begin{figure}[t]
  \centering
  \captionsetup[subfigure]{justification=centering}
  \includegraphics[width=\linewidth]{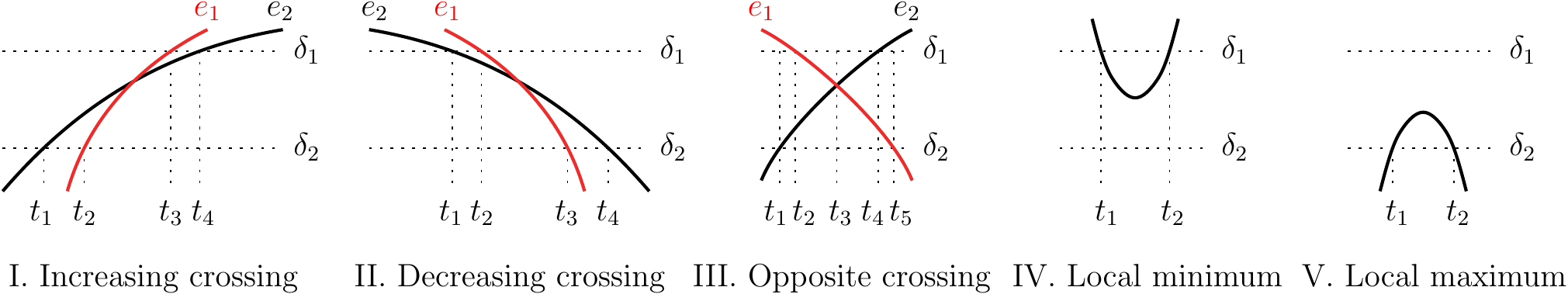}
  \caption{The events that change the zigzag filtration of $\Rcal^\dG$ as $\dG$ varies. 
  Each (partial) distance-time curve corresponds to a vertex pair, and for some events,
  edges formed by the vertex pairs are also denoted.}
  \label{fig:dpc_events}
\end{figure}
\begin{description}
  \item[I.\ Increasing crossing]:
  In Figure~\ref{fig:dpc_events},
  $e_1$ is deleted first at $t_3$ and then $e_2$ is deleted at $t_4$ in $\Rcal^{\dG_1}$.
  In $\Rcal^{\dG_2}$, the deletions of $e_1,e_2$ are switched.
  The switch of edge deletions 
  in the zigzag filtrations
  is realized by a sequence
  of simplex-wise \emph{backward switches}.

  \item[II.\ Decreasing crossing]:
  This is symmetric to the increasing crossing
  where additions of two edges are switched. It is realized by a sequence
  of simplex-wise \emph{forward switches}.

  \item[III.\ Opposite crossing]:
  In Figure~\ref{fig:dpc_events},
  $e_1$ is added first at $t_2$ and then $e_2$ is deleted at $t_3$
  in $\Rcal^{\dG_1}$.
  In $\Rcal^{\dG_2}$, the addition of $e_1$ and the deletion of $e_2$ are switched.
  The simplex-wise version of $\Fcal^{\dG_1}$ contains the following part
  \[R^{\dG_1}_{t_1}\inctosp{\sG_1}\cdots\inctosp{\sG_q} R^{\dG_1}_{t_3}
  \bakinctosp{\tG_1}\cdots\bakinctosp{\tG_r} R^{\dG_1}_{t_5},\]
  where $t_1,t_3,t_5$ are as defined in Figure~\ref{fig:dpc_events}.
  To obtain $\Fcal^{\dG_2}$, we do the following for each $i=1,\ldots,r$:
  \begin{itemize}
      \item 
  If $\tG_i$ is not equal to any of $\sG_1,\ldots,\sG_q$, 
  then use \emph{outward switches} to make $\bakinctosp{\tG_i}$ appear
  immediately before the additions of $\sG_1,\ldots,\sG_q$.
  If $\tG_i$ is equal to a $\sG_j$, first use outward switches 
  to make $\bakinctosp{\tG_i}$ appear
  immediately after $\inctosp{\sG_j}$. 
  Then, apply the \emph{inward contraction} on $\inctosp{\sG_j}\bakinctosp{\tG_i}$.
  Note that $\sG_j$ ($=\tG_i$) which contains both $e_1,e_2$
  does not exist in any complex $R^{\dG_2}_t$
  for $t$ a time shown in Figure~\ref{fig:dpc_events}
  because $e_1,e_2$ do not both exist in these complexes.
  \end{itemize}

  \item[IV.\ Local minimum]:
  In Figure~\ref{fig:dpc_events},
  an edge $e$ corresponding to the black curve 
  is added at $t_1$ and then deleted at $t_2$ in $\Rcal^{\dG_1}$.
  In $\Rcal^{\dG_2}$, the addition and deletion of $e$ disappear.
  Correspondingly,
  simplices containing $e$ are added and then deleted in $\Fcal^{\dG_1}$, 
  but in $\Fcal^{\dG_2}$, the addition and deletion of the above mentioned simplices do not exist.
  Hence, we need to perform \emph{inward contractions}.
  Note that before this, we may need to perform forward or backward switches
  to properly order the additions and deletions.
  (For example, suppose that $\sG$ is the last simplex added 
  due to the addition of $e$.
  However, if $\sG$ is not the first simplex deleted 
  due to the deletion of $e$,
  we need to perform backward switches to make 
  this true
  so that we can perform an inward contraction on $\sG$.)

  \item[V.\ Local maximum]:
  In Figure~\ref{fig:dpc_events},
  an edge $e$ corresponding to the black curve 
  exists in any complex $R^{\dG_1}_t$ for $t$ a time shown in the figure.
  However, in $\Rcal^{\dG_2}$,
  $e$ is deleted at $t_1$ and then added at $t_2$.
  Accordingly, we need to perform \emph{outward expansions}
  on simplices which are deleted and then added.
\end{description}

All the above five types of points appear in Figure~\ref{fig:dis_plot} 
with the numbering of types labelled.

\section{Proof of Proposition~\ref{prop:universality}}
\label{sec:pf-prop-universality}
We prove that any simplex-wise zigzag filtration
as stated in the proposition
can be transformed into an empty filtration
by the update operations in this paper.
This implies that an empty filtration
can be transformed into any simplex-wise filtration
by the reverse operations.
The proposition is then true.

Let
$
\Fcal:
\emptyset=
K_0\leftrightarrowsp{\fsimp{}{0}} K_1\leftrightarrowsp{\fsimp{}{1}}
\cdots 
\leftrightarrowsp{\fsimp{}{\filtcnt-1}} K_\filtcnt
=\emptyset
$
be a simplex-wise zigzag filtration.
We first transform $\Fcal$
into an \emph{up-down}~\cite{carlsson2009zigzag-realvalue} simplex-wise filtration:
\[\Ucal:\emptyset=L_0\hookrightarrow L_1\hookrightarrow
\cdots\hookrightarrow 
L_{\simpcnt}\hookleftarrow
L_{\simpcnt+1}\hookleftarrow
\cdots
\hookleftarrow L_{2\simpcnt}=\emptyset.\]
Let $K_{i}\bakinctosp{\fsimp{}{i}}K_{i+1}$ be the first deletion in $\Fcal$
and $K_{j}\inctosp{\fsimp{}{j}}K_{j+1}$ be the first addition after that.
That is, $\Fcal$ is of the form
\[
\Fcal:K_0\hookrightarrow \cdots\hookrightarrow
K_{i}\bakinctosp{\fsimp{}{i}} 
K_{i+1}\bakinctosp{\fsimp{}{i+1}}
\cdots
\bakinctosp{\fsimp{}{j-2}}
K_{j-1}\bakinctosp{\fsimp{}{j-1}}
K_{j}\inctosp{\fsimp{}{j}}K_{j+1}
\leftrightarrow\cdots \leftrightarrow K_\filtcnt.
\]
If $\fsimp{}{j-1}\neq\fsimp{}{j}$,
we perform inward switch on
$\bakinctosp{\fsimp{}{j-1}}\inctosp{\fsimp{}{j}}$ 
to derive a filtration
\[
K_0\hookrightarrow \cdots\hookrightarrow
K_{i}\bakinctosp{\fsimp{}{i}} 
K_{i+1}\bakinctosp{\fsimp{}{i+1}}
\cdots
\bakinctosp{\fsimp{}{j-2}}
K_{j-1}\inctosp{\fsimp{}{j}}
K'_{j}\bakinctosp{\fsimp{}{j-1}}
K_{j+1}
\leftrightarrow\cdots \leftrightarrow K_\filtcnt.
\]
If $\fsimp{}{j-1}=\fsimp{}{j}$,
we perform outward contraction on $\bakinctosp{\fsimp{}{j-1}}\inctosp{\fsimp{}{j}}$ 
to derive a filtration
\[
K_0\hookrightarrow \cdots\hookrightarrow
K_{i}\bakinctosp{\fsimp{}{i}} 
K_{i+1}\bakinctosp{\fsimp{}{i+1}}
\cdots
\bakinctosp{\fsimp{}{j-2}}
K_{j+1}
\leftrightarrow\cdots \leftrightarrow K_\filtcnt.
\]
We can continue the above operations until
there are no additions after deletions, so that
the filtration becomes an up-down one.

Finally,
on the up-down filtration,
we perform forward/backward switches and inward contractions
to transform it into an empty one.

\section{The remaining update algorithms based on maintaining representatives}
\label{sec:update-alg-cont}

\subsection{Backward switch}
Backward switch is symmetric to forward switch
and hence the algorithm for it is also symmetric.
So, we omit the details.

\subsection{Inward switch}
\label{sec:iwd-switch}

Recall that an inward switch is the following operation:

\vspace{-1em}
\begin{center}
\begin{tikzpicture}
\tikzstyle{every node}=[minimum width=24em]
\node (a) at (0,0) {$\Fcal: K_0 \leftrightarrow
\cdots
\leftrightarrow 
K_{i-1}\bakinctosp{\sG} 
K_i 
\inctosp{\tG} K_{i+1}
\leftrightarrow
\cdots \leftrightarrow K_\filtcnt$}; 
\node (b) at (0,-0.6){$\Fcal': K_0 \leftrightarrow
\cdots
\leftrightarrow 
K_{i-1}\inctosp{\tG} 
K'_i 
\bakinctosp{\sG} K_{i+1}
\leftrightarrow
\cdots \leftrightarrow K_\filtcnt$};
\path[->] (a.0) edge [bend left=90,looseness=1.5,arrows={-latex},dashed] (b.0);
\end{tikzpicture}
\end{center}

\vspace{-1em}\noindent
where $\sG\neq\tG$.
By the Mayer-Vietoris Diamond 
Principle~\cite{carlsson2010zigzag,carlsson2019parametrized,carlsson2009zigzag-realvalue},
there is a bijection between $\Pers_*(\Fcal)$ and $\Pers_*(\Fcal')$.
Let $[b,d]$ be an interval in $\Pers_*(\Fcal)$ with the following representative:
\[\rseq:\chn_{b-1}\dashleftarrow\cyc_{b}\rseqlrarr{\chn_{b}}\cdots
\rseqlrarr{\chn_{d-1}}\cyc_{d}\dashrightarrow\chn_{d}.\]
As in Section~\ref{sec:owd-switch},
we have the following seven cases for $[b,d]\in\Pers_*(\Fcal)$:

\setcounter{desccounter}{0}
\begin{description}
\descitem{Case}{($b=i,d=i$)}:\label{itm:iwd-switch-bi-di}
Suppose that $[b,d]\in\Pers_*(\Fcal)$ is 
in dimension $\Dim$.
The corresponding interval in $\Pers_{*}(\Fcal')$ is also
$[b,d]$ but in dimension $\Dim+1$.
Since $\partial(\chn_{i-1})=\cyc_i=\partial(\chn_{i})$,
we have that $\partial(\chn_{i-1}+\chn_{i})=0$,
which means that $\chn_{i-1}+\chn_{i}$ is a $(\Dim+1)$-cycle in $K'_i$.
Also, since $\sG\in\chn_{i-1}$, $\tG\in\chn_{i}$,
$\tG\not\in\chn_{i-1}$ (because $\tG\not\in K_{i-1}$), 
and $\sG\not\in\chn_{i}$ (because $\sG\not\in K_{i+1}$),
we have that $\sG,\tG\in\chn_{i-1}+\chn_{i}$.
Hence, the representative for $[b,d]\in\Pers_{\Dim+1}(\Fcal')$
consists of only the cycle $\chn_{i-1}+\chn_{i}$.

\descitem{Case}{($b<i,d=i$)}:\label{itm:iwd-switch-bleqi-1-di}
The corresponding interval in $\Pers_{*}(\Fcal')$
is $[b,i-1]$.
We have that $\cyc_{i-1}+\cyc_{i}=\partial(\chn_{i-1})$ and
$\cyc_{i}=\partial(\chn_{i})$
for $\chn_{i-1}\subseteq K_{i-1}\subseteq K'_i$ and
$\chn_{i}\subseteq K_{i+1}\subseteq K'_i$.
So $\cyc_{i-1}=\partial(\chn_{i-1}+\chn_{i})$.
Since $\tG\in\chn_{i}$ and $\tG\not\in\chn_{i-1}$ (because $\tG\not\in K_{i-1}$),
it is true that $\tG\in\chn_{i-1}+\chn_{i}$.
Then, the representative for $[b,i-1]\in\Pers_{*}(\Fcal')$ is:
\[\chn_{b-1}\dashleftarrow\cyc_{b}\rseqlrarr{\chn_{b}}\cdots
\rseqlrarr{\chn_{i-2}}\cyc_{i-1}\dashrightarrow\chn_{i-1}+\chn_{i}.\]

\descitem{Case}{($b=i,d>i$)}:\label{itm:iwd-switch-bi-dgeqi+1}
This case is symmetric to Case~\ref{itm:iwd-switch-bleqi-1-di}
and the details are omitted.

\descitem{Case}{($b<i,d>i$)}:\label{itm:iwd-switch-bleqi-1-dgeqi+1}
The corresponding interval in $\Pers_{*}(\Fcal')$
is still $[b,d]$ and the representative
stays the same besides the changes on the arrow directions.
For example, $\cyc_{i-1}\rseqlarr{\chn_{i-1}}\cyc_{i}$ in $\rseq$
now becomes $\cyc_{i-1}\rseqrarr{\chn_{i-1}}\cyc_{i}$
after the switch, 
where $\chn_{i-1}\subseteq K_{i-1}\subseteq K'_i$
and $\cyc_{i}\subseteq K_{i}\subseteq K'_i$.

\descitem{Case}{($b=i+1$)}:\label{itm:iwd-switch-bi+1}
The corresponding interval in $\Pers_{*}(\Fcal')$
is $[i,d]$ with the following representative:
\[\cyc_{i}\rseqlarr{0}
\cyc_{i+1}\rseqlrarr{\chn_{i+1}}\cdots
\rseqlrarr{\chn_{d-1}}\cyc_{d}\dashrightarrow\chn_{d},
\mbox{where $\cyc_{i}:=\cyc_{i+1}$}.\]

\descitem{Case}{($d=i-1$)}:\label{itm:iwd-switch-di-1}
This case is symmetric to Case~\ref{itm:iwd-switch-bi+1}
and the details are omitted.

\descitem{Case}{($b>i+1$ or $d<i-1$)}:\label{itm:iwd-switch-rest}
The corresponding interval in $\Pers_{*}(\Fcal')$ is $[b,d]$
and the representative
stays the same.

\end{description}

\paragraph{Time complexity.}
Traversing the intervals in $\Pers_*(\Fcal)$
takes $O(\filtcnt)$ time,
and all the cases take no more than $O(\simpcnt)$ time
with Case~\ref{itm:iwd-switch-rest} taking constant time.
Since Cases~\ref{itm:iwd-switch-bi-di}$-$\ref{itm:iwd-switch-di-1} 
execute for only a fixed number of times, 
the time complexity of inward switch operation is $O(\filtcnt)$.

\subsection{Inward expansion}
\label{sec:in-expan}

Recall that an inward expansion is the following operation:

\vspace{-0.5em}
\begin{center}
\begin{tikzpicture}
\node (a) at (0,0) {$\Fcal: K_0 \leftrightarrow
\cdots\leftrightarrow K_{i-2}
\leftrightarrow 
K_i 
\leftrightarrow K_{i+2}\leftrightarrow
\cdots \leftrightarrow K_\filtcnt$}; 
\node (b) at (0,-0.6){$\Fcal': K_0 \leftrightarrow
\cdots\leftrightarrow K_{i-2}
\leftrightarrow 
K'_{i-1}
\inctosp{\sG} 
K'_i 
\bakinctosp{\sG} 
K'_{i+1}
\leftrightarrow K_{i+2}\leftrightarrow
\cdots \leftrightarrow K_\filtcnt$};
\draw[->,arrows={-latex},dashed] (a.0) .. controls (+6.5,+0) and (+7,-0.35) .. (b.0);
\end{tikzpicture}
\end{center}

\vspace{-1em}\noindent
where $K'_{i-1}=K_{i}=K'_{i+1}$.
We also assume that $\sG$ is a $\Dim$-simplex.
Note that 
indices for $\Fcal$ are nonconsecutive
in which $i-1$ and $i+1$ are skipped.

For the update,
we first determine whether the induced map 
$\Hm_*(K'_{i-1})\rightarrow \Hm_*(K'_i)$
is injective
or 
surjective
by checking whether 
$\partial(\sG)$ is a $(\Dim-1)$-boundary in $K'_{i-1}$
(injective)
or not
(surjective).
The checking can be done by 
performing a reduction of $\partial(\sG)$
on a $(\Dim-1)$-boundary basis for $K'_{i-1}$
(which can be computed by a persistence algorithm~\cite{cohen2006vines}).

\subsubsection{$\Hm_*(K'_{i-1})\rightarrow \Hm_*(K'_i)$ is injective}
\label{sec:in-expan-inj}
The only difference of $\Pers_*(\Fcal)$ and $\Pers_*(\Fcal')$
in this case 
is that 
there is a new interval
$[i,i]$ in $\Pers_*(\Fcal')$.
The representative $\Dim$-cycle 
at index $i$ for $[i,i]\in\Pers_*(\Fcal')$
can be any $\Dim$-cycle in $K'_i$ containing $\sG$,
which can be computed from the reduction done previously
on $\partial(\sG)$ and the $(\Dim-1)$-boundary basis for $K'_{i-1}$.
Also, any interval $[b,d]\in\Pers_*(\Fcal)$
is an interval in $\Pers_*(\Fcal')$;
the update of representative for $[b,d]\in\Pers_*(\Fcal')$
is as in Section~\ref{sec:out-expan-surj}.

\subsubsection{$\Hm_*(K'_{i-1})\rightarrow \Hm_*(K'_i)$ is surjective}

In this case, 
$\pinds(\Fcal')=\pinds(\Fcal)\union\Set{i+1}$ and
$\ninds(\Fcal')=\ninds(\Fcal)\union\Set{i-1}$.
Let $\Set{I_j\given j\in\Bcal}$ be the set of intervals in $\Pers_{\Dim-1}(\Fcal)$
containing $i$, where $\Bcal$ is an indexing set.
Also, let $\tilde{\cyc}^j_{i}$ be the representative $(\Dim-1)$-cycle
at index $i$ for $I_j$.
We have that the homology classes 
$\Set{[\tilde{\cyc}^j_{i}]\given {j\in\Bcal}}$
form a basis for $\Hm_{\Dim-1}(K_{i})=\Hm_{\Dim-1}(K'_{i-1})$.
Denote the map $\Hm_*(K'_{i-1})\rightarrow \Hm_*(K'_i)$ 
as $\rho$;
then, there exists a non-empty set $\LG\subseteq\Bcal$ 
s.t.\ $\sum_{j\in\LG}[\tilde{\cyc}^j_{i}]\in\ker(\rho)$.
The set $\LG$ can be computed by forming a
$(\Dim-1)$-cycle basis for $K'_{i-1}$ by combining
$\Set{\tilde{\cyc}^j_{i}\given {j\in\Bcal}}$
with the $(\Dim-1)$-boundary basis for $K'_{i-1}$,
and then performing a Gaussian elimination and reduction.

We then rewrite the intervals in $\Set{I_j\given j\in\LG}$ as 
\[[b_1,d_1],[b_2,d_2],\ldots,[b_\ell,d_\ell]\text{ s.t.\ }b_1\bles b_2\bles\cdots\bles b_\ell.\]

For each $j$ s.t.\ $1\leq j\leq\ell$, 
let $\rseq_j$ denote the representative sequence for $[b_j,d_j]\in\Pers_*(\Fcal)$,
and let $\cyc^j_i$ denote the $(\Dim-1)$-cycle at index $i$ in $\rseq_j$.
We then pair the birth indices $i+1,b_1,\ldots,b_\ell$
with the death indices $i-1,d_1,\ldots,d_\ell$ to form intervals
for $\Pers_*(\Fcal')$.
We first pair $b_\ell$ with $i-1$ to form an interval $[b_\ell,i-1]\in\Pers_*(\Fcal')$,
whose representative is
derived from $\rseq_1\rsprefix{i}\repsum\cdots\repsum\rseq_\ell\rsprefix{i}$.
The representative for $[b_\ell,i-1]\in\Pers_*(\Fcal')$ is
valid because: (i) $\sum_{j=1}^\ell[{\cyc}^j_{i}]\in\ker(\rho)$;
(ii) $b_\ell=\max_{\bles}\Set{b_1,\ldots,b_\ell}$.
Symmetrically,
we pair $i+1$ with $d_*=\max_{\dles}\Set{d_1,\ldots,d_\ell}$ 
to form an interval $[i+1,d_*]\in\Pers_*(\Fcal')$,
whose representative is
derived from $\rseq_1\rssuffix{i}\repsum\cdots\repsum\rseq_\ell\rssuffix{i}$.

Then,
we pair the 
remaining indices.
Specifically,
for $r:=1,\ldots,\ell-1$,
pair $b_r$ with a death index as follows:
\begin{itemize}
    \item If $d_r$ is unpaired, then pair $b_r$ with $d_r$.
    The representative for $[b_r,d_r]\in\Pers_*(\Fcal')$
    can be updated from the representative for $[b_r,d_r]\in\Pers_*(\Fcal)$
    as described in Section~\ref{sec:in-expan-inj}.
    
    \item If $d_r$ is paired, then $d_1,\ldots,d_r$
    must be all the paired death indices among $d_1,\ldots,d_\ell$ so far.
    Since $d_{r+1},\ldots,d_\ell$ are all unpaired, 
    we pair $b_r$ with $\dG=\Max_{\dles}\Set{d_{r+1},\ldots,d_\ell}$.
    We then describe how we obtain the representative for $[b_r,\dG]\in\Pers_*(\Fcal')$.
    For each $j$ s.t.\ $1\leq j\leq r$,
    we define the following representative $\tilde{\rseq}_j$ for $[b_j,i]$
    in $\Fcal'$:
    first take the representative sequence $\rseq_j\rsprefix{i}$ in $\Fcal$
    and treat it as a representative sequence for $[b_j,i-1]$ in $\Fcal'$;
    then attach a cycle at index $i$ to $\rseq_j\rsprefix{i}$ 
    by copying the cycle at index $i-1$,
    to derive $\tilde{\rseq}_j$
    (note that $\rseq_j\rsprefix{i}$ is treated as a representative in $\Fcal'$
    and hence the last index is $i-1$).
    Symmetrically,
    for each $j$ s.t.\ $r<j\leq\ell$,
    we define the representative $\tilde{\rseq}_j$ for $[i,d_j]$
    in $\Fcal'$,
    which is derived from $\rseq_j\rssuffix{i}$.
    With the above definitions,
    the representative for $[b_r,\dG]\in\Pers_*(\Fcal')$ is the following:
    \begin{equation*}
    \big(\tilde{\rseq}_1\brepsum
    \cdots\brepsum\tilde{\rseq}_r\big)\rconcat
    \big(\tilde{\rseq}_{r+1}\drepsum
    \cdots\drepsum\tilde{\rseq}_\ell\big).
    \end{equation*}
    The concatenation in the above representative is well-defined because 
    $\sum_{j=1}^\ell[{\cyc}^j_{i}]=0$ in $K'_i$,
    which means that $\sum_{j=1}^r[{\cyc}^j_{i}]=\sum_{j=r-1}^\ell[{\cyc}^j_{i}]$.
\end{itemize}

Finally,
all remaining intervals in $\Pers_*(\Fcal)$
are carried into $\Pers_*(\Fcal')$;
the update of representatives for these intervals
is the same as in Section~\ref{sec:in-expan-inj}.

\subsubsection{Time complexity}
The inward expansion operation takes $O(\filtcnt\simpcnt^2)$ time.
The analysis is similar to the analysis for outward contraction 
in Section~\ref{sec:out-contra-complexity} but is easier.

\subsection{Inward contraction}
\label{sec:in-contra}

Recall that an inward contraction is the following operation:

\vspace{-1em}
\begin{center}
\begin{tikzpicture}
\node (a) at (0,0) {$\Fcal: K_0 \leftrightarrow
\cdots\leftrightarrow K_{i-2}
\leftrightarrow 
K_{i-1}\inctosp{\sG} 
K_i 
\bakinctosp{\sG} K_{i+1}
\leftrightarrow K_{i+2}\leftrightarrow
\cdots \leftrightarrow K_\filtcnt$}; 
\node (b) at (0,-0.6){$\Fcal': K_0 \leftrightarrow
\cdots\leftrightarrow K_{i-2}
\leftrightarrow 
K'_{i}
\leftrightarrow K_{i+2}\leftrightarrow
\cdots \leftrightarrow K_\filtcnt$};
\draw[->,arrows={-latex},dashed] (a.0) .. controls (+7,-0.2) and (+6.5,-0.6) .. (b.0);
\end{tikzpicture}
\end{center}

\vspace{-1em}\noindent
where $K'_{i}=K_{i-1}=K_{i+1}$.
We also assume that $\sG$ is a $\Dim$-simplex.
Note that indices for $\Fcal'$ are not consecutive,
i.e., $i-1$ and $i+1$ are skipped.

For the update,
we first determine whether the induced map 
$\Hm_*(K_{i-1})\rightarrow \Hm_*(K_i)$
is injective
or 
surjective
by checking whether 
$i$ is a birth index in $\Fcal$
(injective)
or $i-1$ is a death index in $\Fcal$
(surjective).

\subsubsection{$\Hm_*(K_{i-1})\rightarrow \Hm_*(K_i)$ is injective}
\label{sec:in-contra-inj}
Since inward contractions are
inverses of inward expansions (see Section~\ref{sec:in-expan}),
the only difference of $\Pers_*(\Fcal)$ and $\Pers_*(\Fcal')$
in this case 
is that 
$[i,i]\in\Pers_*(\Fcal)$ is deleted in $\Pers_*(\Fcal')$.

Let $[b,d]\neq [i,i]$ be an interval in $\Pers_*(\Fcal)$.
If $i\not\in[b,d]$, i.e., $b>i$ or $d<i$, then since $b\neq i+1$
and $d\neq i-1$, we have that $b\geq i+2$ or $d\leq i-2$.
So the representative for $[b,d]\in\Pers_*(\Fcal)$
can be directly used as a representative for $[b,d]\in\Pers_*(\Fcal')$.

If $i\in[b,d]$,  
let $\tilde{\cyc}_i$ be the representative $\Dim$-cycle at index $i$ 
for $[i,i]\in\Pers_*(\Fcal)$,
and let
\[\rseq:\chn_{b-1}\dashleftarrow\cyc_{b}\rseqlrarr{\chn_{b}}\cdots
\cdots
\rseqlrarr{\chn_{d-1}}\cyc_{d}\dashrightarrow\chn_{d}\]
be the representative sequence for $[b,d]\in\Pers_*(\Fcal)$.
Note that $\sG\in\tilde{\cyc}_i\subseteq K_i$.
Since $\cyc_{i-1}+\cyc_{i}=\partial(\chn_{i-1})$
and $\cyc_{i}+\cyc_{i+1}=\partial(\chn_{i})$
for $\chn_{i-1},\chn_{i}\subseteq K_i$,
we have that
$\cyc_{i-1}+\cyc_{i+1}=\partial(\chn_{i-1}+\chn_{i})$
for $\chn_{i-1}+\chn_{i}\subseteq K_i$.
If $\sG\not\in\chn_{i-1}+\chn_{i}$,
then $\chn_{i-1}+\chn_{i}\subseteq K'_i$.
If $\sG\in\chn_{i-1}+\chn_{i}$, 
we say that $\rseq$ is \emph{$\sG$-relevant}.
We have that 
$\partial(\chn_{i-1}+\chn_{i}+\tilde{\cyc}_i)=\partial(\chn_{i-1}+\chn_{i})=\cyc_{i-1}+\cyc_{i+1}$,
where $\chn_{i-1}+\chn_{i}+\tilde{\cyc}_i$ does not contain $\sG$
and hence is in $K'_i$.
So we always have that $\cyc_{i-1}+\cyc_{i+1}=\partial(\bar{\chn})$
for a chain $\bar{\chn}\subseteq K'_i$.
Then, the representative for $[b,d]\in\Pers_*(\Fcal')$
is set as:
\begin{equation*}
\chn_{b-1}\dashleftarrow\cyc_{b}\rseqlrarr{\chn_{b}}\cdots
\rseqlrarr{\chn_{i-3}}\cyc_{i-2}
\rseqlrarr{\chn_{i-2}}\cyc'_{i}
\rseqlrarr{\bar{\chn}+\chn_{i+1}}\cyc_{i+2}
\rseqlrarr{\chn_{i+2}}
\cdots
\rseqlrarr{\chn_{d-1}}\cyc_{d}\dashrightarrow\chn_{d},
\end{equation*}
where $\cyc'_{i}:=\cyc_{i-1}$.

\subsubsection{$\Hm_*(K_{i-1})\rightarrow \Hm_*(K_i)$ is surjective}
\label{sec:in-contra-sur}
In this case, $i-1\in\ninds(\Fcal)$, $i+1\in\pinds(\Fcal)$,
$\ninds(\Fcal')=\ninds(\Fcal)\setminus\Set{i-1}$,
and $\pinds(\Fcal')=\pinds(\Fcal)\setminus\Set{i+1}$.
Let $\Set{[\bG_j,\dG_j]\given j\in\Bcal}$ be the set of 
intervals in $\Pers_{*}(\Fcal)$ containing $i$, where $\Bcal$
is an indexing set,
and let $\tilde{\rseq}_j$ be the representative sequence
for each $[\bG_j,\dG_j]$.
Moreover, define a set $\LG\subseteq\Bcal$ as: 
\[\LG:=\Set{j\in\Bcal\given \tilde{\rseq}_j\text{ is }\sG\text{-relevant}}.\]

We do the following:
\begin{itemize}
    \item Whenever there exist $j,k\in\LG$
s.t.\ $[\bG_j,\dG_j]\iles [\bG_k,\dG_k]$, 
update the representative for $[\bG_k,\dG_k]$ 
as $\tilde{\rseq}_j\repsum{\tilde{\rseq}}_k$,
and delete $k$ from $\LG$.
Note that $\tilde{\rseq}_j\repsum{\tilde{\rseq}}_k$
is \emph{$\sG$-irrelevant}.
\end{itemize}

After the above operations,
we have that no two intervals in 
$\Set{[\bG_j,\dG_j]\given j\in\LG}$ are comparable. 
Let $[b_*,i-1]$ and $[i+1,d_\circ]$ be the $(\Dim-1)$-th intervals in $\Pers_*(\Fcal)$
ending/starting with $i-1$, $i+1$ respectively.
Moreover, let $\rseq_*$ be the representative sequence
for $[b_*,i-1]$, and
let $\rseq_\circ$ be the representative sequence
for $[i+1,d_\circ]$.
We do the following:
\begin{itemize}
    \item Whenever there is a $j\in\LG$
s.t.\ $b_*\bles\bG_j$, 
update the representative for $[\bG_j,\dG_j]$ 
as $\rseq_*\repsum{\tilde{\rseq}}_j$,
and delete $j$ from $\LG$.
Note that $i-1\dles\dG_j$ and $\rseq_*\repsum{\tilde{\rseq}}_j$
is $\sG$-irrelevant.

\item Whenever there is a $j\in\LG$
s.t.\ $d_\circ\dles\dG_j$, 
update the representative for $[\bG_j,\dG_j]$ 
as $\rseq_\circ\repsum{\tilde{\rseq}}_j$,
and delete $j$ from $\LG$.
Note that $i+1\bles\bG_j$ and $\rseq_\circ\repsum{\tilde{\rseq}}_j$
is $\sG$-irrelevant.
\end{itemize}

After the above operations,
we have that $\bG_j\bles b_*$
and $\dG_j\dles d_\circ$ for each $j\in\LG$.
If $\LG=\emptyset$,
then let $[b,d]$ form an interval in 
$\Pers_*(\Fcal')$ with a representative $\rseq_*\rconcat\rseq_\circ$.
If $\LG\neq\emptyset$,
then rewrite the intervals in $\Set{[\bG_j,\dG_j]\given j\in\LG}$ as:
\[[b_1,d_1],[b_2,d_2],\ldots,[b_\ell,d_\ell]\text{ s.t.\ }b_1\bles b_2\bles\cdots\bles b_\ell.\]

Also, for each $j$, let $\rseq_j$ be the $\Dim$-th representative sequence
for $[b_j,d_j]\in\Pers_*(\Fcal)$.

For $j\leftarrow 1,\ldots,\ell-1$, we do the following:
\begin{itemize}
    \item Note that $d_{j+1}\dles d_j$
because otherwise $[b_j,d_j]$ and $[b_{j+1},d_{j+1}]$ would be comparable.
Then, let $[b_{j+1},d_j]$ form an interval in $\Pers_*(\Fcal')$.
The representative is set as follows:
since $\rseq_j\repsum\rseq_{j+1}$ 
is a representative for $[b_{j+1},d_j]$
in $\Fcal$ which is $\sG$-irrelevant,
$\rseq_j\repsum\rseq_{j+1}$ can be `contracted' to become 
a representative for $[b_{j+1},d_j]\in\Pers_*(\Fcal')$
as done in Section~\ref{sec:in-contra-inj}.
\end{itemize}

We then do the following:
\begin{itemize}
    \item Let $[b_*,d_\ell]$ form an interval in $\Pers_*(\Fcal')$
    whose representative is derived from $\rseq_*\repsum\rseq_\ell$
    (which is $\sG$-irrelevant);
    let $[b_1,d_\circ]$ form an interval in $\Pers_*(\Fcal')$
    whose representative is derived from $\rseq_\circ\repsum\rseq_1$
    (which is $\sG$-irrelevant).
\end{itemize}

Finally, for each remaining interval $[b,d]\in\Pers_*(\Fcal)$,
whose representative is $\sG$-irrelevant,
$[b,d]$ forms an interval in $\Pers_*(\Fcal')$,
whose representative is updated as in Section~\ref{sec:in-contra-inj}.

\subsubsection{Time complexity}
By a similar analysis as in Section~\ref{sec:out-contra-complexity},
the total time spent on the injective case is $O(\filtcnt\simpcnt)$.
The bottleneck of the surjective case is the loops,
which take $O(\filtcnt\simpcnt^2)$ time.
Hence, the inward contraction operation takes $O(\filtcnt\simpcnt^2)$ time.

\end{document}